%% file: cohe_context.tex
\documentclass[showkeys,floatfix,noshowpacs, preprintnumbers, twocolumn, aps,superscriptaddress,pra,longbibliography,10pt]{revtex4-1}

\input{macros.tex}

\newcommand{\psection}[1]{\paragraph*{#1.---}\hspace{-1em}}

\begin{document}

\newcommand{\inl}{INL -- International Iberian Nanotechnology Laboratory, Av. Mestre Jos\'{e} Veiga s/n, 4715-330 Braga, Portugal}
\newcommand{\inlshort}{INL -- International Iberian Nanotechnology Laboratory, Braga, Portugal}
\newcommand{\uff}{Instituto de F\'{i}sica, Universidade Federal Fluminense, Av. Gal. Milton Tavares de Souza s/n, Niter\'{o}i -- RJ, 24210-340, Brazil}
\newcommand{\uffshort}{Instituto de F\'{i}sica, Universidade Federal Fluminense, Niter\'{o}i -- RJ, Brazil}
\newcommand{\cfum}{Centro de F\'{i}sica, Universidade do Minho, Campus de Gualtar, 4710-057 Braga, Portugal}
\newcommand{\cfumshort}{Centro de F\'{i}sica, Universidade do Minho, Braga, Portugal}

\title{Inequalities witnessing coherence, nonlocality, and contextuality}

\author{Rafael Wagner}
\email{rafael.wagner@inl.int}
\affiliation{\inlshort}
\affiliation{\cfumshort}
\author{Rui Soares Barbosa}
\email{rui.soaresbarbosa@inl.int}
\affiliation{\inlshort}
\author{Ernesto F.~Galvão}
\email{ernesto.galvao@inl.int}
\affiliation{\inlshort}
\affiliation{\uffshort}

\date{\today}

\begin{abstract}
Quantum coherence, nonlocality, and contextuality are key resources for quantum advantage in metrology, communication, and computation.
We introduce a graph-based approach to derive classicality inequalities that bound local, noncontextual, and coherence-free models,
offering a unified description of these seemingly disparate quantum resources.
Our approach generalizes recently proposed basis-independent coherence witnesses, and recovers all noncontextuality inequalities of the exclusivity graph approach.
Moreover,  violations of certain classicality inequalities witness preparation contextuality. We describe an algorithm to find all such classicality inequalities, and use it to analyze some of the simplest scenarios.
\end{abstract}

\maketitle

\psection{Introduction}%
Non-classical resources provided by quantum theory are key to quantum advantage for information processing \cite{HowardWVE14,Abbot12,Raussendorf13,Abramsky17,Mansfield18,Saha19,lostaglio2020certifying,Kirby20}; see \cite{Budroni21,Brunner14,Streltsov17} for comprehensive reviews of applications.
Many different nonclassical features of quantum mechanics have been identified, studied, witnessed, and quantified
\cite{Winter16,Selby20,Amaral19,Duarte18,Wolfe20,Schmid20resource,KangDaWu21,Chitambar19,Abramsky17,Abramsky19,Barbosa23,wagner2021using,regula2017convex,regula2022tight,theurer2017resource,designolle2021set,regula2022probabilistic}.
It is natural to wonder to what extent different quantum resources can be characterized in a unified way.
Here we address this question by proposing a single formalism that yields inequalities bounding three different notions of classicality: noncontextual, local, and coherence-free models.

A number of modern approaches to contextuality have successfully incorporated nonlocality as a special case \cite{cabello2014graph,AbramskyB11,acin2015combinatorial,amaral2018graph}.
The relationship between this unified notion of non-classical correlations and coherence, however, has been harder to establish.
One roadblock is that most approaches to characterize coherence presuppose the choice of a fixed reference basis \cite{Streltsov17}.
Recently, different approaches have been proposed to study a basis-independent notion of coherence \cite{GalvaoB20, designolle2021set}, dubbed \stress{set coherence} in Ref.~\cite{designolle2021set}.
A recent approach, on which the present work builds, derives witnesses of basis-independent coherence using only relational information between states in the form of two-state overlaps~\cite{GalvaoB20}.
Still, so far there has been no clear identification between non-locality and contextuality on one hand, and coherence on the other.
There are examples of models that mimic quantum coherence \textit{without} displaying contextuality or nonlocality, such as the toy models from Refs.~\cite{spekkens2007evidence,catani2021interference}, while on the other hand incoherent states -- even maximally mixed states -- can of course be used to witness state-independent quantum contextuality~\cite{cabello2008experimentally,amselem2009state}.
Theory-independent approaches have been used to compare relevant types of nonclassical resources~\cite{takagi2019general,regula2017convex,regula2022probabilistic}, but an understanding of the special case of coherence and contextuality is still lacking.
A better understanding of the relationship between these two fundamental manifestations of nonclassicality 
has both important foundational impact and potential technological applications.

Building on the study of coherence using two-state overlaps \cite{GalvaoB20},
we propose a framework that associates to any (simple) graph $G$ a probability polytope $C_G$ of edge weightings.
Vertices of the graph $G$ represent probabilistic processes, while edges of $G$ correspond to correlations between neighbouring processes.
We show that the faces of the polytope $C_G$ describe bounds on noncontextual, local, and coherence-free models,
depending on the interpretation of vertices of the graph $G$ as preparations and measurements.
The description of three notions of classicality under a single framework represents a significant conceptual advance towards clarifying the source of quantum computational advantage.

\psection{The classical polytope $C_G$}%
Let $G = (V(G), E(G))$ be an undirected graph, which we call the \stress{event graph}.
We consider edge weightings $r \colon E(G) \to [0,1]$,
which assign a weight $r_e = r_{ij}$ to each edge $e=\{i,j\}$ of $G$.
We regard these weightings as points forming a polytope, the unit hypercube, $r \in [0,1]^{E(G)}$.
To define the \stress{classical polytope} $C_G \subseteq [0,1]^{E(G)}$,
take each vertex $i \in V(G)$ to represent a random variable $A_i$ with values belonging to an alphabet $\Lambda$,
and suppose these are jointly distributed.
This determines an edge weighting $r$ where each weight $r_{ij}$ is the probability that the processes corresponding to vertices $i$ and $j$ output equal values, \ie \[r_{ij} = P(A_i = A_j) .\]
An edge weighting $r$ is in the classical polytope $C_G$ if it arises in this fashion from jointly distributed random variables $(A_i)_{i \in V(G)}$.
Each weight $r_{ij}$ is then a measure of the correlation between the output values of $A_i$ and of $A_j$.
In the case of dichotomic values $\Lambda = \{ +1, -1\}$, this quantity is related to the expected value of the product by  $\left\langle A_i A_j\right\rangle = 2r_{ij}-1$ \footnote{Note that we do not assume a fixed finite outcome set $\Lambda$, or a bound on its size. The classical polytope consists of the edge weightings that arise from jointly distributed random variables with outcomes in \stress{some} set $\Lambda$. We could fix a single $\Lambda$ as long as it is countably infinite. But in practice, for a fixed graph $G$ with $n$ vertices, it suffices to consider $\Lambda = \enset{1, \ldots, n}$.}.
An (alternative) formal description of $C_G$ is given in detail in \Cref{app:algo}.

\psection{Inequalities defining $C_G$}%
The inequalities defining the polytope $C_G$ impose logical conditions determining the set of classical edge weightings.
The existence of non-trivial facets of $C_G$ can be illustrated with the example of \cref{fig:all_graphs_together}--(a), the 3-vertex complete graph $K_3$, with edge weights $r_{12}, r_{23}, r_{13}$.
We cannot have e.g.\
\[r_{12} = 1, \quad r_{23}=1, \quad r_{13}=0, \]
as this would contradict transitivity of equality on the deterministic values corresponding to each of the three vertices: $A_1 = A_2 = A_3 \neq A_1$.
In Ref.~\cite{GalvaoB20} it was shown that the only non-trivial inequalities for the $n$-cycle event graph $C_n$ are
\begin{equation}
    -r_{e} +\sum_{{e'} \neq {e}} r_{e'} \leq n-2,\;\; \text{for each $e \in E(C_n)$.} \label{eq:ncycle}
\end{equation}
Incidentally, these inequalities have been known at least since the work of Boole \cite{Boole1854, Pitowsky94,Abramsky2012logical}.

We now give a high-level description of an algorithm to completely characterize $C_G$ for general event graphs $G$.
We start by enumerating the vertices of $C_G$.
These are all the `deterministic' labellings of the edges of $G$ with values in $\{0,1\}$
that are logically consistent with transitivity of equality.
The facets of $C_G$ can then be found using standard convex geometry tools~\footnote{The inequalities found in this work were obtained using the \textbf{traf} option from the \textbf{PORTA} program, which converts a V-representation of a polytope into an H-representation.}.

Whether a given deterministic edge labelling is consistent -- and therefore a vertex of $C_G$ -- can be checked in linear time on the size of $G$ by a graph traversal.
However, it is unnecessary to generate all $2^{\lvert E(G) \rvert}$-many labellings and discard the inconsistent ones.
Instead, one can directly generate only the consistent ones by searching through underlying value assignments to the vertices of $G$.
Despite being much more efficient for most graphs, this also quickly becomes unavoidably intractable due to the exponentially-increasing number of vertices of the polytope $C_G$.
We deepen this discussion in \Cref{app:algo}.

Using the method just outlined, we find all facets of $C_G$ for some small graphs, including all graphs shown in \cref{fig:all_graphs_together}. 
Interestingly, already for $K_4$ (\cref{fig:all_graphs_together}--(c)) a new type of facet appears which is different from the cycle inequalities in \cref{eq:ncycle}.
These new facets of $K_4$ are described by the inequalities of the form
\begin{equation}
    (r_{12}+r_{13}+r_{14}) - (r_{23}+r_{34}+r_{24}) \le 1, \label{eq:k4nontrivial}
\end{equation}
(and others obtained by label permutations).

In \Cref{app:buildgraphs}, we prove that some constructions of graphs by combining smaller graphs do not give rise to new facet inequalities, trimming the class of graphs worth analyzing.
In \Cref{app:smallgraphs}, we list all facet inequalities of the classical polytopes for the complete graphs $K_4$, $K_5$, and $K_6$.
We also give numerically-found examples of quantum violations -- witnessing basis-independent coherence in the sense described in the next section --
of all non-trivial facets of $K_4$ and $K_5$.
All the new inequalities and quantum violations found, together with the code used to obtain them, which is applicable to analyze $C_G$ for an arbitrary graph $G$, are made available in an associated Git repository \cite{wagner2022github}.
In \Cref{app:family}, we generalize the inequalities of \cref{eq:k4nontrivial} 
to complete graphs of $n \ge 4$ vertices, and prove that these define facets of the classical polytopes $C_{K_n}$ for all such $n$.
This yields an infinite family $h_n$ of new classicality inequalities not previously described in the literature.
The first three new inequalities from this family  ($h_4$, $h_5$, $h_6$) have recently been experimentally violated, serving to benchmark quantum photonic devices \cite{giordani23experimental}.

We now proceed to describe how the inequalities obtained for the abstract scenarios considered above establish bounds both on coherence-free models and on noncontextual/local models.
Each type of operational scenario suggests an interpretation for edge weights, and naturally imposes further constraints on them, resulting in cross-sections of the polytope $C_G$.
These cross-sections recover known noncontextuality/locality polytopes, as well as basis-independent coherence witnesses.

\psection{$C_G$ bounds coherence-free models}%
Most commonly, coherence is defined for a quantum state with respect to a fixed basis, as the presence of nonzero off-diagonal elements in its density matrix (in that basis)~\cite{baumgratz2014quantifying,aberg2006quantifying}.
Recently, Refs.~\cite{GalvaoB20,designolle2021set} proposed a \stress{basis-independent} notion of coherence as a property of a \stress{set of states}:
this is said to be \stress{coherent} when the states in the set are not simultaneously diagonalizable, \ie when there is no basis in which all their density matrices are diagonal,
or equivalently, if the states in the set do not pairwise commute.
Otherwise, the set is said to be \stress{coherence-free}, or \stress{incoherent}.

In Ref.~\cite{GalvaoB20}, basis-independent coherence witnesses were described using only pairwise overlaps $r_{ij}=\Tr(\rho_i \rho_j)$ among a set of quantum states, focusing on witnesses provided by violations of the cycle inequalities in \cref{eq:ncycle}. 
We explain the interpretation of the facet inequalities of $C_G$ as basis-independent coherence witnesses, generalizing the results of Ref.~\cite{GalvaoB20} to \stress{any} event graph $G$.

Let $G$ be any graph with $n$ vertices.
Consider a general separable state of $n$ quantum systems of the same type (e.g. qudits), 
each associated to a vertex of the graph.
Each edge of $G$ is given a weight equal to the overlap between the two states of its incident vertices.
These overlaps can be estimated using the well-known SWAP test \cite{BuhrmanCWdW01}.
In Ref.~\cite{GalvaoB20} it was shown that the facet-inequalities of $C_G$ describe necessary conditions on the set of overlaps,
\ie on edge weightings of $G$, 
for the set of single-system states to be coherence-free, that is,
all of them diagonal in a common single-system basis.
This is so because for such a coherence-free set of states,
the overlap $r_{ij}$ equals the probability of obtaining equal outcomes in independent measurements of the states associated to vertices $i$ and $j$ using the observable that projects onto the reference basis.

\psection{$C_G$ bounds local and noncontextual models}%
The faces of $C_G$ can also be understood as bounds on noncontextual models~\cite{KochenS67,Bell64}.
A simple first approach consists in having vertices of $G$ represent measurements, while edges identify two-measurement contexts, i.e. pairs of observables that can be measured simultaneously.
The weight of an edge corresponds to the probability, with respect to a given global state, that the two incident measurements yield equal outcomes.
A necessary and sufficient condition for the existence of a noncontextual model whose behaviour is consistent with a given edge weighting
is the existence of a global probability distribution (on outcome assignments to all measurements) whose marginals recover the correct outcome probabilities.
This is the content of the Fine--Abramsky--Brandenburger theorem \cite{Fine82PRL,Fine82JMP,AbramskyB11}.

Such a global distribution, when it exists, can also be interpreted as a classical coherence-free model.
This dual role of global probability distributions is the link connecting coherence-free models and noncontextual models, and allowing violations of facet inequalities of $C_G$ to witness either property, depending on the interpretation of the scenario at hand.

In general, this simple approach,
interpreting vertices as measurements and edges as equality of outcome in two-measurement contexts,
is not sufficient to capture contextuality in full generality~\cite[Section 2.5.3]{amaral2018graph}.
Even restricting to contextuality scenarios whose maximal contexts have size two, the facets of $C_G$ are not necessarily facet, or even tight, noncontextuality inequalities,
except in the case of dichotomic measurements~\cite[Theorem 38]{araujo2014quantum}, where equality of outcomes fully determines the measurement statistics. An important example is the Clauser--Horne--Shimony--Holt (CHSH) inequality.

\begin{figure*}[thb]
    \centering
    \includegraphics[width=0.8\textwidth]{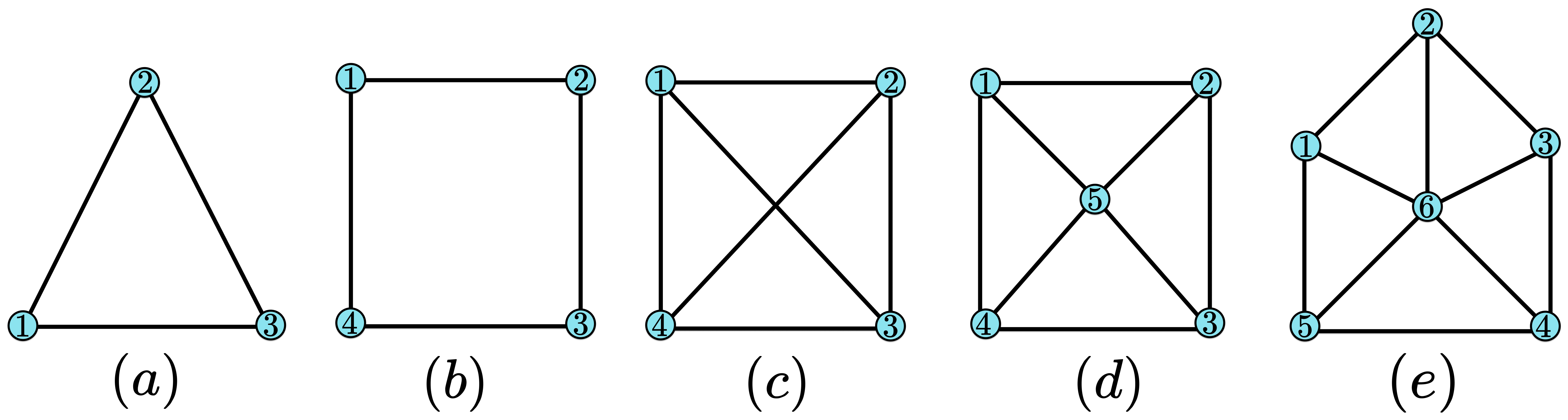}
    \caption{\textbf{Event graphs corresponding to bounds on classical models.} Each of these graphs can be used to obtain the following nonclassicality inequalities: (a) constrained CHSH inequality; (b),(d) CHSH Bell locality inequality; (c) new $K_4$ classicality inequality from \cref{eq:k4nontrivial}, and (e) Klyachko, Can, Binicio\u{g}lu, and Schumovsky (KCBS) noncontextuality inequality.}
    \label{fig:all_graphs_together}
\end{figure*}

Encoding some contextuality scenarios requires the imposition of further constraints, which geometrically determine cross-sections on the classical polytope $C_G$. These constraints may, for example, represent operational symmetries of the measurement scenario, e.g. making two vertices equal, or may encode given conditions on the compatibility of observables. One example is the exclusivity constraint present in the Cabello--Severini--Winter (CSW) graph approach~\cite{cabello2014graph}.

We now show how both CHSH and the original 3-setting Bell inequality can be obtained from cycle inequalities, before describing a more systematic approach that recovers all noncontextuality inequalities obtainable from the exclusivity graph approach~\cite{cabello2014graph,amaral2018graph}. 

We remark that we treat Bell nonlocality as an instance of contextuality, in which measurement compatibility is ensured by space-like separation between various parties who locally measure a shared multipartite system.
This view of nonlocality as a special case of contextuality is well established, e.g., in Refs.~\cite{AbramskyB11,acin2015combinatorial},
although there are important subtle differences when considering free transformations in a resource-theoretic setup~\cite{karvonen2021neither}.

\psection{Example: CHSH inequality from the 4-cycle graph $C_4$}%
It is easy to check from \cref{eq:ncycle} that the 4-cycle graph $C_4$ with edges $r_{12}, r_{23}, r_{34}, r_{14}$ (see \cref{fig:all_graphs_together}--(b)) has $4$ non-trivial facets given by the inequality
\begin{equation}
    r_{12}+r_{23}+r_{34}-r_{14} \le 2, \label{eq:4cycle}
\end{equation}
and label permutations thereof.
We translate this into the CHSH~\cite{ClauserHSH69} Bell scenario, with Alice locally measuring one of two rank-1 projectors $A_1$ or $A_2$, and Bob locally measuring either $B_1$ or $B_2$, on the singlet state $\ket{\psi}=\frac{1}{\sqrt{2}} (\ket{01}-\ket{10})$.
As a contextuality scenario, the CHSH graph $C_4$ is
a graph with no clique with more than two vertices,
and the only non-trivial noncontextuality inequality is given in terms of correlations.
From the event graph perspective, each vertex can be understood as a two-outcome measurement at either Alice or Bob.
It is easy to check that the overlap between two single-qubit projectors $A$, $B$ is the probability of obtaining different outcomes~\footnote{Equal outcomes and different outcomes are described in a dual way. Without loss of generality, we can use either the probability that the two measurements return equal outcomes or the probability that they return different outcomes. The forbidden deterministic edge labellings are essentially the same up to permuting $0$ and $1$.}
when measuring those projectors on each part of the singlet state: $r_{AB}=p_{\neq}^{AB}=1-p_{=}^{AB}$.
Using this interpretation, the facet of $C_{C_4}$ given by \cref{eq:4cycle} can be rewritten as
\begin{equation}
    \label{eq:eqchsh}
    p_{\neq}^{{A_1}{B_1}}+p_{\neq}^{{A_2}{B_1}}+p_{\neq}^{{A_2}{B_2}}-p_{\neq}^{{A_1}{B_2}} \le 2,
\end{equation}
which is a well-known way to write the CHSH inequality~\cite{collins2002bell}. This same procedure can be used to obtain chained Bell inequalities \cite{braunstein1990wringing,araujo2013all} from cycle inequalities.

\psection{Example: Original Bell inequality from the 3-cycle graph $C_3$}%
If on the $C_4$ graph we have just analyzed we impose the constraint that one of the edge weights equal 1,
we recover the non-trivial facets for the 3-cycle $C_3$, namely $r_{12}+r_{23}-r_{13} \le 1$ and label permutations.
The embedded tetrahedron with these 3 facets delimits the local correlations in the original two-party Bell inequality \cite{Bell64}, featuring three settings at each party, and assuming perfect anticorrelation for pairs of aligned settings. For a geometrical description of the elliptope of quantum correlations, see Ref.~\cite{Le2023quantumcorrelations}.

\psection{Example: CHSH inequality from the 5-vertex wheel graph $W_5$}%
An alternative way of interpreting an event graph as a contextuality scenario involves having a single vertex, the handle, represent a quantum state, and all the others represent measurement operators.
Take the 5-vertex wheel graph $W_5$ of \cref{fig:all_graphs_together}--(d) as an instructive example.
A simple calculation shows that if we impose $r_{12}=r_{34}$ and $r_{23}=r_{14}$,
then adding together four $3$-cycle inequalities for this graph recovers the CHSH inequality in the form of \cref{eq:eqchsh}.
The quantum realization of this graph scenario has the central vertex $5$ representing a singlet state, with the other vertices representing the four projectors measured jointly by Alice and Bob.
The imposed constraints reflect the fact that opposing edges represent the same quantity, the overlap between the two projectors locally measured by one of the parties.

\psection{Recovering all noncontextuality inequalities of the exclusivity graph formalism}%
The second approach to obtaining the CHSH inequality does not rely on particular properties of the singlet state.
The use of a handle vertex to represent a state can be generalized to other scenarios, as we now describe.

In the \stress{exclusivity graph approach} to contextuality one considers a graph $H$ whose vertices represent measurement events\footnote{One may think of a \stress{measurement event} as a pair $(m,o)$ describing that the measurement $m$ is performed and the outcome $o$ is observed.} (in a quantum realization, projection operators),
and edges connect mutually exclusive events (in the quantum setting, orthogonal projectors).
In this formalism, the noncontextual behaviours are described by a well-known construction, the stable polytope of the graph $H$, denoted $\STAB(H)$~\cite{amaral2018graph}.
This is reviewed in detail in \Cref{app:contextuality}.
In brief, the vertices or extreme points of the polytope $\STAB(H)$ are (the characteristic functions of) subsets of $V(H)$ that do not contain any pair of adjacent vertices.
More intuitively, perhaps, they correspond to truth-value assignments to the measurement events,
\ie functions $V(H) \to \enset{0,1}$,
such that no two exclusive events are deemed true,
\ie no two adjacent vertices are assigned the value $1$.

We can understand this setup in terms of our formalism as follows.
We define an event graph $H_\star$ obtained from the exclusivity graph $H$ by adding a new vertex connected to all other vertices.
This new vertex is used to represent a handle state $\psi$.
Formally, $H_\star$ is given by
$V(H_\star) \defeq V(H) \sqcup \{\psi\}$ and $E(H_\star) \defeq E(H)\cup \setdef{\{v,\psi\}}{v\in V(H)}$.
The structure of the exclusivity graph $H$ is then used to force a constraint on edge weightings of $H_\star$,
namely that all edges already present in $H$ be assigned zero weight.
The resulting cross-section
$C_{H_\star}^0 \defeq \setdef{r \in C_{H_\star}}{\Forall{e\in E(H)} r_e=0}$
of the polytope $C_{H_\star}$, which moreover is a subpolytope, then carries information about the noncontextual behaviours in $\STAB(H)$.
Formally, in \Cref{app:contextuality}, we exhibit an isomorphism between the polytopes $\STAB(H)$ and $C_{H_{\star}}^0$ for any exclusivity graph $H$.
As a consequence, we show that the \textit{facet-defining noncontextuality inequalities bounding noncontextual behaviours for $H$ are precisely the facet-defining inequalities of $C_{H_\star}^0$}.
Moreover, these inequalities can be obtained from the inequalities defining facets of the whole classical polytope $C_{H_\star}$ by removing (\ie setting to zero) the variables $r_e$ with $e \in E(H)$.

\psection{Example: KCBS noncontextuality inequality}%
We illustrate this mapping between formalisms with the noncontextuality inequality obtained by Klyachko, Can, Binicio\u{g}lu, and Schumovsky (KCBS) \cite{KlyachkoCBS08}, and expressed in the CSW formalism in Ref.~\cite{cabello2014graph}.

Starting with the 5-cycle graph $H = C_5$ interpreted as an exclusivity graph, then $H_*$ is
the 6-vertex wheel graph $W_6$ of \cref{fig:all_graphs_together}--(e).
The central vertex represents a quantum state,
while neighbouring vertices in the outer 5-cycle represent mutually exclusive measurement events
(quantum mechanically: orthogonal projectors)
so as to impose $r_{vw}=0$ for neighbouring $v$ and $w$ in this outer subgraph.
The KCBS noncontextuality inequality is a bound on weightings of the edges connected to the central vertex:
\begin{equation}\label{eq: KCBS event graph}
    \sum_{v=1}^{5} r_{v6} \le 2 .
\end{equation}
Note that each edge weight $r_{v6}$ in \cref{eq: KCBS event graph} is the probability of successful projection of the central vertex state onto the projector associated with vertex $v$.

In our framework, this inequality is obtained from a facet-defining inequality of $C_{W_6}$,
\[
-r_{12}-r_{23}-r_{34}-r_{45}-r_{15}+r_{16}+r_{26}+r_{36}+r_{46}+r_{56} \leq 2,
\]
by imposing the exclusivity (or orthogonality) condition of null edge weights on the 5-cycle outer subgraph.

\psection{Cycle inequalities witness preparation contextuality}%
Besides considering different approaches to Kochen--Specker noncontextuality, one can also consider different \textit{notions} of noncontextuality.
One such proposal, put forth by Spekkens in Ref.~\cite{spekkens2005contextuality}, is that of preparation (generalized) noncontextuality~\cite{Lostaglio2020contextualadvantage,schmid2018discrimination,Spekkens08,baldijao_emergence_2021,lostaglio2020certifying}.
We consider once more a quantum realization of the event graph representing vertices as states and edges as two-state overlaps.
In \Cref{app:preparationcontextuality} we prove that \textit{violations of the inequalities for the classical polytope of the cycle event graph $C_n$ are witnesses of preparation contextuality}.
This result is shown for a class of prepare-and-measure operational scenarios~\cite{Lostaglio2020contextualadvantage,schmid2018discrimination},
which includes quantum theory viewed as an operational theory.
In contrast to quantum theory, the well-known noncontextual toy theory of Ref.~\cite{spekkens2007evidence} does not violate these event graph inequalities,
if vertices of the event graph are taken to represent toy theory states. 

\psection{Discussion and future directions}%
We proposed a new graph-theoretic approach that unifies the study of three different quantum resources, namely contextuality, nonlocality, and coherence.
Non-classicality inequalities are obtained as facets of a polytope $C_G$ of edge weightings associated with an \textit{event graph} $G$, with suitable constraints that depend on the chosen interpretation of vertices as quantum states or measurements, as required by each scenario.  

Connections with the theory of contextuality were presented with respect to different approaches and definitions.
In particular, we recovered all inequalities of the CSW exclusivity graph approach \cite{cabello2014graph},
and we explicitly derived CHSH and KCBS inequalities as examples.
We also showed that for cycle graphs the classical polytope bounds Spekkens preparation noncontextuality. 

It would be interesting to understand whether these results can be made more robust.
In particular, we observed that the noncontextuality inequalities for exclusivity graphs $H$ are obtained from the inequalities of a classical polytope $C_{H_\star}$ by assigning weight zero to some edges. But many of these inequalities of $C_{H_\star}$ \textit{allow} for deviations from such null weights without leaving the classical polytope $C_{H_\star}$.
This suggests that perhaps those inequalities could still be interpreted as a robust form of noncontextuality inequalities, where exclusivity is relaxed.

Future research directions include characterizing this framework in the landscape of general probabilistic theories (GPTs) and understanding how this approach bounds relational unitary invariants involving three or more states, such as Bargmann invariants \cite{Oszmaniec21}.
It would also be interesting to relate violation of our inequalities with advantage in quantum protocols, as recently done by some of us in \cite{WagnerCG22} for the task of quantum interrogation.

\psection{Acknowledgements}%
We would like to thank Marcelo Terra Cunha, John Selby, David Schmid, and Raman Choudhary for helpful discussions.
We also thank Roberto D.~Baldijão for critically reviewing an early version of this work.

We acknowledge financial support from FCT -- Funda\c{c}\~ao para a Ci\^encia e a Tecnologia (Portugal) through PhD Grant SFRH/BD/151199/2021 (RW) and through CEECINST/00062/2018 (RSB and EFG).
This work was supported by the ERC Advanced Grant QU-BOSS, GA no. 884676.

\appendix

\section{Characterizing the classical polytope}\label{app:algo}

In Ref.~\cite{GalvaoB20}, Galv\~{a}o and Brod derived the facet-defining inequalities of the classical polytope $C_{C_n}$ of the $n$-cycle event graph $C_n$, as discussed in the text.
The construction uses an argument based on Boole's inequalities for logically consistent processes~\cite{Boole1854}.
In the main text we discuss that, in fact, \textit{any} event graph, and not only cycle graphs, can be used to bound classicality of different forms.

In this section, we consider the computational problem of characterizing the classical polytope $C_G$ for any event graph $G$.
We propose a simple algorithm for computing all its vertices and facets.
This proceeds by first calculating the list of vertices of $C_G$, \ie its V-representation,
and then finding its facet-defining inequalities, \ie its H-representation, using standard convex geometry tools.
As discussed in the main text, this last step is computationally efficient on the size of the polytope.
However, the overall efficiency of the procedure is intrinsically limited by the fact that the number of vertices and facets of $C_G$ grows exponentially on the size of $G$.
The brunt of this section is dedicated to computing the set of vertices of $C_G$.

After setting out the formal definitions, we characterize the edge $\enset{0,1}$-labellings
$E(G) \to \enset{0,1}$ that respect logical consistency conditions and thus correspond to the vertices of $C_G$.
This characterization yields an efficient procedure for checking whether such an edge labelling is a vertex of $C_G$, whose complexity we analyze.

However, when the goal is to generate all vertices of $C_G$, it is needlessly wasteful to generate all the $2^{|E(G)|}$-many edge $\{0,1\}$-labellings and then filter them one by one.
Instead, we present a procedure that generates the edge labellings that are vertices of $C_G$ by generating vertex labellings underlying them, thus limiting the search through the space $\enset{0,1}^{E(G)}$ of edge labellings.
Even though it might output the same vertex more than once, the method works well, especially for dense graphs. It is optimal for the complete graphs $K_n$, which as we will see in \Cref{app:smallgraphs} are our main examples of interest.
We observe that the number of vertices of the polytope $C_{K_n}$ is given by a well-known combinatorial sequence, known as the Bell numbers \cite{OEISBell}, which count the number of partitions of a set, precisely the space that is searched by this procedure.
Finally, we discuss an alternative method that might be more efficient for sparse graphs.

\psection{Basic definitions}%
We start with the relevant definitions.

\begin{definition}[Graph]\label{def: graph}
A \textit{graph} $G = (V(G), E(G))$
consists of a finite set $V(G)$ of vertices
and a set $E(G)$ of edges,
which are two-element subsets of $V(G)$, \ie sets of the form $\enset{v,w}$ where $v, w \in V(G)$ are distinct vertices.
\end{definition}
Note that the graphs we consider in this text are so-called \textit{simple} graphs:
they are undirected (since $\enset{v,w} = \enset{w,v}$), have at most one edge between any two vertices $v$ and $w$, and have no loops (i.e. have no edges from a vertex to itself).
In one well-delimited passage, however, we will need to consider \textit{possibly loopy graphs},
which may have loops. This corresponds to dropping the requirement that $v$ and $w$ be distinct in the definition above. A possibly loopy graph is said to be \textit{loop-free} if it has no loops, \ie if is is a bona fide (simple) graph.

\begin{definition}[Labellings and colouring]\label{def: vertex labelling and colouring}
A \textit{vertex labelling} by a set $\Lambda$,
or a \textit{vertex $\Lambda$-labelling} for short, is a function $\fdec{\lambda}{V(G)}{\Lambda}$ assigning to each vertex a label from $\Lambda$.
It is called a \textit{colouring} if $\{v,w\} \in E(G)$ implies $\lambda(v) \neq \lambda(w)$.
The graph $G$ is said to be $k$-colourable for $k \in \NN$ when it admits a colouring by a set of size $k$.

Similarly, an \textit{edge labelling}
by a set $\Lambda$, or an \textit{edge $\Lambda$-labelling} for short, is a function $\fdec{\alpha}{E(G)}{\Lambda}$ 
assigning a label from $\Lambda$ to each edge. When $\Lambda=[0,1]$, we call this an \textit{edge weighting}.
\end{definition}

\begin{definition}[Chromatic number]
The \textit{chromatic number} of a graph $G$, written $\chi(G)$, is the smallest $k \in \NN$ such that $G$ is $k$-colourable.
\end{definition}

In the classical, deterministic situation modelled by our framework, we consider a vertex labelling of a graph $G$ by an arbitrary labelling set $\Lambda$.
However, operationally, we do not have access to the vertex labels, but only to the information
of whether the labels of neighbouring edges are equal or different.
\begin{definition}
Given any vertex labelling
$\fdec{\lambda}{V(G)}{\Lambda}$,
its \textit{equality labelling}
$\epsilon_\lambda$ is
the edge $\enset{0,1}$-labelling 
given by:
\begin{align*}
    \fdec{\epsilon_\lambda&}{E(G)}{\enset{0,1}}
    \\
\epsilon_\lambda& \,\{v,w\} \,\defeq\,
\delta_{\lambda(v),\lambda(w)} = \begin{cases} 1 & \text{if $\lambda(v)=\lambda(w)$} \\ 0 & \text{if $\lambda(u) \neq \lambda(v)$}\end{cases}
\Mdot
\end{align*}
\end{definition}

We are interested in characterizing the edge $\enset{0,1}$-labellings that arise as equality labellings of vertex labellings.
\begin{definition}
An edge $\enset{0,1}$-labelling $\fdec{\alpha}{E(G)}{\enset{0,1}}$ is said to be
\textit{$\Lambda$-realizable}
if it is the equality labelling of some vertex $\Lambda$-labelling,
\ie if $\alpha = \epsilon_\lambda$ for some $\fdec{\lambda}{V(G)}{\Lambda}$.
If $\Lambda$ has size $k \in \NN$, we say that $\alpha$
is $k$-realizable.
\end{definition}

We write $\EqG$ for the set of realizable edge labellings of $G$ (with any $\Lambda$),
and $\EqGk$ for the set of $k$-realizable ones.
We have that $\EqGk \subseteq \Eqk{G}{k'}$ whenever $k \leq k'$,
and $\EqG = \cup_{k \in\NN}\EqGk$.
Moreover,
$\EqG = \Eqk{G}{|V(G)|}$ because a vertex labelling uses at most one distinct label per vertex of the graph.

We often refer to these realisable edge $\enset{0,1}$-labellings as the \textit{classical} edge labellings. 
By the inclusion $\{0,1\}\subseteq[0,1]$, we can think of any edge $\enset{0,1}$-labelling as a (deterministic) edge weighting.
This gives an alternative description of the classical polytope $C_G$ in the main text.

\begin{definition}
Given a graph $G$, its \textit{classical polytope} $C_G \subseteq [0,1]^{E(G)}$ is the convex hull of the set $\EqG$ seen as a set of points in $[0,1]^{E(G)}$.
\end{definition}

\psection{Characterizing the vertices of $C_G$}%
We now consider the question of determining whether a given
edge $\enset{0,1}$-labelling is realizable (as the equality labelling of some vertex labelling).

Given $\fdec{\alpha}{E(G)}{\enset{0,1}}$, define a relation
$\sim_\alpha$ on the set of vertices of $G$ whereby
$v \sim_\alpha w$ if and only if there is a path from $v$ to $w$ through edges labelled by $1$,
\ie
there is a sequence $e_1, \ldots, e_n \in E(G)$ such that $v \in e_1$, $w \in e_n$, $e_i \cap e_{i+1} \neq \es$, and $\alpha(e_i)=1$.
This is easily seen to be an equivalence relation.

It yields the following characterization of the classical edge labellings.
\begin{proposition}\label{prop:charclassicalvertices}
An edge labelling $\fdec{\alpha}{E(G)}{\enset{0,1}}$ is realizable (\ie classical)
if and only if for all edges $\enset{v,w} \in E(G)$, $v \sim_\alpha w$ implies $\alpha(\enset{v,w})=1$.
\end{proposition}
In other words,  an edge labelling $\fdec{\alpha}{E(G)}{\enset{0,1}}$ fails to be realizable
precisely when 
there is an edge $\enset{v,w} \in E(G)$ such that $v \sim_\alpha w$ and $\alpha(\enset{v,w})=0$.
In terms of the underlying vertex labellings, such a situation would violate the transitivity of equality.

A slightly different perspective is given by using $\alpha$ to construct a new graph that `collapses' $G$ through paths labelled by $1$. Note that this construction yields a possibly loopy graph.

An edge $\enset{0,1}$-labelling $\alpha$ partitions the edges of $G$ into two sets.
This determines two graphs
 $G_{\alpha = 0}$ and $G_{\alpha = 1}$,
both with the same vertex set as $G$, but each retaining only the edges of $G$ with the corresponding label,
i.e. for each $\lambda \in \enset{0,1}$, 
\begin{align*}
    V(G_{\alpha = \lambda}) &\defeq V(G)\\
    E(G_{\alpha = \lambda}) &\defeq \setdef{e \in E(G)}{\alpha(e) = \lambda}
\end{align*}
A possibly loopy graph $G/\alpha$ is then defined as follows:
\begin{itemize}
    \item its vertices are connected components of $G_{\alpha = 1}$, or equivalently, the equivalence classes of $\sim_\alpha$;
    \item there is an edge between two connected components $A$ and $B$ of $G_{\alpha = 1}$ whenever there exist vertices $v \in A$, $w \in B$, such that $\{v,w\} \in E(G_{\alpha=0})$.
\end{itemize}

\begin{lemma}
Let $\fdec{\alpha}{E(G)}{\enset{0,1}}$ and $\Lambda$ be any set.
There is a one-to-one correspondence between $\Lambda$-realizations of $\alpha$ and $\Lambda$-colourings of $\Galpha$.
\end{lemma}
\begin{proof}
Let $\fdec{\lambda}{V(G)}{\Lambda}$ such that $\alpha = \epsilon_\lambda$.
If $v \sim_\alpha w$, then $\lambda(v)=\lambda(w)$, by propagating equality along the path labelled by $1$.
Hence, the map $\fdec{\kappa}{V(\Galpha)}{\Lambda}$ given by $\kappa([v])\defeq\lambda(v)$ is well defined.
Now, an edge $e \in E_\Galpha$
is of the form $e=\enset{[v],[w]}$ for some
$v, w \in V(G)$ such that $\alpha(\enset{v,w}) = 0$.
Since $\alpha = \epsilon_\lambda$, this means that $\lambda(v) \neq \lambda(w)$, hence $\kappa([v]) \neq \kappa([w])$. Thus, $\kappa$ is a colouring.

Conversely, given a colouring $\fdec{\kappa}{V_\Galpha}{\Lambda}$, set
$\lambda(v) \defeq \kappa([v])$. Let $e = \enset{v, w} \in E(G)$. 
If $\alpha(e)=1$, then $[v]=[w]$, hence $\lambda(v)=\lambda(w)$ because $\kappa$ is a colouring.
If $\alpha(e)=0$, then $\enset{[v],[w]}\in E_\Galpha$, hence $\lambda(v)\neq\lambda(w)$. In either case, $\alpha(e) = \epsilon_\lambda(e)$.

The two processes just described are inverses of one another.
\end{proof}

\begin{corollary}\label{cor:realizablecolourable}
An edge $\enset{0,1}$-labelling is $\Lambda$-realizable if and only if the possibly loopy graph $G/\alpha$ is $\Lambda$-colourable.
In particular, it is realizable (\ie classical) if and only if $G/\alpha$ is loop-free.
\end{corollary}

\begin{proposition}
Checking whether an edge $\enset{0,1}$-labelling for a graph $G$ is realizable can be done in time $O(n+m)$ where $n = |V(G)|$ and $m = |E(G)|$.
Checking $k$-realizability in a given $k\geq 3$ is $NP$-complete.
\end{proposition}
\begin{proof}
For the first part, transverse the graph $G_{\alpha=1}$ using a depth-first search (DFS).
When visting each vertex, run through all the departing edges of $G_{\alpha=0}$ to see if any is linked to an already visited vertex in the connected component of $G_{\alpha=1}$ currently being traversed. If any is found, reject $\alpha$.

For the second part, use \cref{cor:realizablecolourable} to reduce to graph colouring:
a graph $G$ is $k$-colourable if and only if the constant $0$ edge labelling is realizable.
\end{proof}

The procedure outlined in the proof above is described below in more detail using pseudo-code.

\vspace{-10pt}
{\scriptsize
\begin{align*}
&\rule{\columnwidth}{.7pt} \\
&\textbf{Input: } \text{graph $G$ with } V(G) = \enset{1, \ldots, N}. \hspace{100cm}\\[-2pt]
&\hphantom{\textbf{Input: }} \text{edge-labelling } \fdec{\alpha}{E(G)}{\enset{0,1}} \\[-2pt]
&\textbf{Output: } \text{whether $\alpha$ is realizable, hence a vertex of the polytope $C_G$}. \\[-2pt]
&\\[-2pt]
&\textbf{global variable } d_i \textbf{ for each } i \in V(G) \\[-2pt]
&\textbf{global variable } c_i \textbf{ for each } i \in V(G) \\[-2pt]
& \\[-2pt]
&\textbf{procedure } \textsc{Main} ()\\[-2pt]
&\algtab d_i \leftarrow \textrm{false} \textbf{ for all } i \in V(G)\\[-2pt]
&\algtab \textbf{for } i \in V(G) \textbf{ do} \\[-2pt]
&\algtab\algtab \textbf{if } \lnot d_i \textbf{ then} \\[-2pt]
&\algtab\algtab\algtab c_j \leftarrow \textrm{false} \textbf{ for all } j \in V(G) \\[-2pt]
&\algtab\algtab\algtab \textsc{Search }(i) \\[-2pt]
&\algtab\algtab \textbf{end if} \\[-2pt]
&\algtab \textbf{end for} \\[-2pt]
&\algtab \textbf{terminate with output } \textrm{true}\\[-2pt]
&\\[-2pt]
&\textbf{procedure } \textsc{Search}(i) \\[-2pt]
&\algtab d_i, c_i \leftarrow \textrm{true} \\[-2pt]
&\algtab \textbf{for } j \in \textsc{Neighbours }(i) \textbf{ do} \\[-2pt]
&\algtab\algtab \textbf{if } \alpha(\enset{i,j})=0 \;\land\; c_j \textbf{ then}\\[-2pt]
&\algtab\algtab\algtab \textbf{terminate with output } \textrm{false}\\[-2pt]
&\algtab\algtab \textbf{else if } \alpha(\enset{i,j}=1) \;\land\; \lnot d_j \textbf{ then}\\[-2pt]
&\algtab\algtab\algtab \textsc{Search }(j) \\[-2pt]
&\algtab\algtab \textbf{end if} \\[-2pt]
&\algtab \textbf{end for} \\[-2pt]
&\vspace{-10pt}\rule{\columnwidth}{.7pt} 
\end{align*}
}
\vspace{-10pt}

\psection{Computing all the vertices of $C_G$}%
We conclude that it is computationally easy to check whether a given edge $\{0,1\}$-labelling,
\ie a given deterministic edge weighting, is classical.
Nevertheless, determining the whole set of vertices of the classical polytope
is computationally hard since the number of edge labelling to be tested grows exponentially with the number of edges of the graph.

It is interesting to note that for the complete event graph $K_n$ of $n$ vertices
the number of classical edge labellings, \ie vertices of the classical polytope $C_{K_n}$,
is given by a well-known sequence, the Bell or exponential numbers \cite{OEISBell,bell1934exponential}.
The $n$-th Bell number is the number of partitions, or equivalence relations, of a set of size $n$.
It is clear that edge $\enset{0,1}$-labellings of $K_n$ are in one-to-one correspondence with symmetric reflexive relations on the set of vertices $\enset{1, \ldots, n}$,
where the label of an edge $\enset{v,w}$ determines whether the pairs $(v,w)$ and $(w,v)$ are in the relation. Among these, the classical edge labellings correspond to the equivalence relations (which additionally satisfy transitivity), with the underlying vertex labelling determining a partition of the vertices.
For a general graph $G$, it is still true that the classical edge labellings arise from partitions, or equivalence relations, on the set of vertices, determined by the underlying vertex labelling. The difference is that an edge labelling does not carry enough information to characterize a relation fully.
So, in particular, different vertex partitions may give rise to the same classical edge labelling.

We can use this observation to propose a different method for generating all vertices of $C_G$ by
constructing vertex-labellings of $G$.
The procedure is given below in pseudo-code.

\vspace{-10pt}
{\scriptsize
\begin{align*}
&\rule{\columnwidth}{.7pt} \\
&\textbf{Input: } \text{graph $G$ with } V(G) = \enset{1, \ldots, N}. \hspace{100cm}\\[-2pt]
&\textbf{Output: } \text{vertices of the polytope $C_G$}. \\[-2pt]
&\\[-2pt]
&\textbf{global variable } \lambda_i \textbf{ for each } i \in V(G) \\[-2pt]
&\textbf{global variable } \alpha_e \textbf{ for each } e \in E(G) \\[-2pt]
& \\[-2pt]
&\textbf{procedure } \textsc{Main} ()\\[-2pt]
&\algtab\textsc{Generate }(1,1)\\[-2pt]
&\textbf{end procedure} \\[-2pt]
&\\[-2pt]
&\textbf{procedure } \textsc{Generate }(i,next) \\[-2pt]
&\algtab\textbf{if } i=N+1 \textbf{ then} \\[-2pt]
&\algtab\algtab \textbf{output } (\alpha_e)_{e \in E(G)} \\[-2pt]
&\algtab\textbf{else} \\[-2pt]
&\algtab\algtab \textbf{for } x < next \textbf{ do} \\[-2pt]
&\algtab\algtab\algtab \textsc{Update } (i,x) \\[-2pt]
&\algtab\algtab\algtab \textsc{Generate } (i+1,next) \\[-2pt]
&\algtab\algtab \textbf{end for} \\[-2pt]
&\algtab\algtab \textsc{Update } (i,next) \\[-2pt]
&\algtab\algtab \textsc{Generate } (i+1,next+1) \\[-2pt]
&\algtab\textbf{end if} \\[-2pt] 
&\textbf{end procedure} \\[-2pt]
&\\[-2pt]
&\textbf{procedure } \textsc{Update }(i,x) \\[-2pt]
&\algtab \lambda_i \leftarrow x \\[-2pt]
&\algtab \textbf{for } j < i \textbf{ with } \enset{i,j} \in E(G) \textbf{ do} \\[-2pt]
&\algtab \algtab \alpha_{\enset{i,j}} \leftarrow \textbf{if } \lambda_j = x \textbf{ then } 1 \textbf{ else } 0 \\[-2pt]
&\algtab \textbf{end for} \\[-2pt]
&\textbf{end procedure}\\[-2pt]
&\vspace{-10pt}\rule{\columnwidth}{.7pt} 
\end{align*}
}\vspace{-10pt}

The procedure above has the disadvantage that it might output the same vertex of the polytope multiple times.
This is because, as already discussed, different partitions of the vertices of $G$ can give rise to the same edge labelling.
The problem is especially noticeable for sparse graphs.

An alternative method for generating the vertices of $C_G$, which might be more efficient in the case of sparser graphs,
is to directly search through $\enset{0,1}^{E(G)}$ while checking for consistency on the fly, in order to trim the search space so that only the realizable edge labellings are constructed.
This can be done by keeping a representation of the current vertex partition (induced by the edges labelled $1$ in the edge labelling being constructed),
for example using a union-find data structure,
together with a record of forbidden merges between partition components (induced by the edges labelled $0$s in the edge labelling being constructed).
The disadvantage is that the upkeep of this representation, necessary for checking consistency on the fly,
cannot be done in constant time. This incurs an overhead at each step in the search.

\section{Characterizing classical polytopes \\ by graph decompositions}\label{app:buildgraphs}

In this section, we prove some general facts that relate the classical polytopes of different graphs.
In particular, we show that some methods of combining graphs to build larger graphs do not give rise to new classicality inequalities.
Or, seen analytically rather than synthetically, that some graphs $G$ can be decomposed into smaller component graphs in a way that reduces the question of characterizing $C_G$ to that of characterizing the polytopes of these components.
These observations help trim down the class of graphs that is worth analyzing in the search for new classicality inequalities.
As a by-product, we characterize the class of graphs for which
all edge weightings are classical as being that of trees, an analogue of Vorob{\textquotesingle}ev's theorem \cite{vorobev1962consistent} in this framework.

\begin{proposition}\label{prop:disjointunion}
Let $G_1$ and $G_2$ be graphs, and write $G_1 + G_2$ for their disjoint union. Then
\[C_{G_1 + G_2} = C_{G_1} \times C_{G_2} = \setdef{(r_1,r_2)}{r_1 \in G_1, r_2 \in G_2}.\]
\end{proposition}
\begin{proof}
Given vertex labellings $\fdec{\lambda_i}{V(G_i)}{\Lambda_i}$ for each $i=1,2$,
one obtains a function 
\[\fdec{\lambda_1+\lambda_2}{V(G_1) \sqcup V(G_2)}{\Lambda_1 \sqcup \Lambda_2}\]
which is a vertex labelling of $G_1+G_2$ since $V(G_1+G_2) = V(G_1)\sqcup V(G_2)$.
The corresponding equality edge labelling,
$\fdec{\epsilon_{\lambda_1+\lambda_2}}{E(G_1+G_2)}{\enset{0,1}}$,
is precisely the function
\[
 \fdec{[\epsilon_{\lambda_1},\epsilon_{\lambda_2}]}{E(G_1) \sqcup E(G_2)}{\enset{0,1}}
\]
given by
\[
 e \longmapsto
 \begin{cases}
    \epsilon_{\lambda_1}(e) & \text{ if $e \in E(G_1)$}
     \\
     \epsilon_{\lambda_2}(e) & \text{ if $e \in E(G_2)$}
 \end{cases},
\]
implying the result.
\end{proof}

In particular, vertices of the polytope $C_{G_1 + G_2}$ are in bijective correspondence with pairs consisting of one vertex from each of the polytopes $C_{G_i}$,
while the facets of $C_{G_1 + G_2}$ are in bijective correspondence with the union of the facets of $C_{G_1}$ and the facets of $C_{G_2}$.
That is, the inequalities defining $C_{G_1+G_2}$ are those defining $C_{G_1}$ plus those defining $C_{G_2}$.
Taking the disjoint union of event graphs thus creates no new classicality inequalities.
As a consequence, we might as well focus solely on studying the classical polytopes of connected graphs.

The result above considers the construction of a new graph by placing two graphs side by side.
But similar results can be obtained for more complicated ways of combining graphs,
namely gluing along a vertex or along an edge.

\begin{definition}[Gluing]
Given graphs $G_1$ and $G_2$, and tuples of vertices
\begin{align*}
\vc{v}_1&=(v_1^1, \ldots, v_1^k) \in V(G_1)^k,\\
\vc{v}_2&=(v_2^1, \ldots, v_2^k) \in V(G_2)^k,
\end{align*}
the \stress{gluing of $G_1$ and $G_2$ along $\vc{v}_1$ and $\vc{v}_2$},
written $G_1 +_{\vc{v}_1=\vc{v}_2} G_2$, is the graph obtained by taking the disjoint union $G_1+G_2$
and identifying the vertices $v_1^j$ and $v_2^j$ for $j = 1,\ldots, k$.
Explicitly: its vertices are
\[
    V(G_1 +_{\vc{v}_1=\vc{v}_2} G_2) \;\defeq\; O_1 \sqcup O_2 \sqcup N ,
\]   
where  $O_i \defeq  V(G_i) \setminus \{v_i^1, \ldots, v_i^k\}$  is the set of vertices of $G_i$ not being identified and
$N = \enset{v^1, \ldots v^k}$ is a set of `new' vertices (\ie not in either $G_i$);
its edges are
\[
E(G_1 +_{\vc{v}_1=\vc{v}_2} G_2) \;\defeq\; E_1 \cup E_2 ,
\]
where $E_i$ is equal to $E(G_i)$ but with occurrences of $v_i^j$ replaced by the new $v^j$.
\end{definition}

\begin{proposition}\label{prop:glue_vertex}
Let $G_1$ and $G_2$ be graphs, $v_1 \in V(G_1)$ and $v_2 \in V(G_2)$,
then $C_{G_1 +_{v_1=v_2} G_2} = C_{G_1} \times C_{G_2}$.
\end{proposition}
\begin{proof}
We proceed as in the proof of \cref{prop:disjointunion}, using the same notation,
but then take a quotient of the merged alphabet $\Lambda_1 \sqcup \Lambda_2$
identifying two labels, one from each component: $\lambda_1(v_1) \in \Lambda_1$ with $\lambda_2(v_2) \in \Lambda_2$.
This yields a well-defined labelling for $G_1 +_{v_1=v_2} G_2$ where the new vertex $v$ is labelled by the element resulting from this identification.
This does not affect the equality edge-labellings, and so we obtain the same result.
\end{proof}

Read analytically, if $G$ is a graph with a cut vertex $V$, \ie a vertex whose removal disconnects the graph into two components with vertex sets $V_1$ and $V_2$, then its polytope can be characterized in terms of the polytopes of the induced subgraph on $V_1 \cup \enset{v}$ and $V_2 \cup \enset{v}$.
In particular, the facet-defining inequalities of $C_G$ are those of each of these two components.

As an aside, this result is the missing ingredient for fully characterizing the event graphs
that cannot display any nonclassicality, \ie 
for which all edge weightings $E(G) \to [0,1]$ are classical.
This could be seen as an analogue of Vorob{\textquotesingle}ev's \cite{vorobev1962consistent} theorem in our framework.
\begin{corollary}
A graph $G$ is such that $C_G = [0,1]^{E(G)}$ if and only if it is a tree.
\end{corollary}
\begin{proof}
For `only if' part, if $G$ has a cycle then any edge labelling $E(G) \to \enset{0,1}$ that restricts to $(1,\ldots,1,0)$ on said cycle is not in $C_G$.
For the `if' part,
apply \cref{prop:glue_vertex} multiple times, following the construction of a tree as a sequence of gluings along a vertex of copies of $K_2$, whose classical polytope is $[0,1]$.
\end{proof}

\begingroup
\squeezetable
\begin{table*}[thb]
\caption{Quantum violations for facet inequalities of $C_{K_5}$\label{tab: inequalities K5}} 
\begin{ruledtabular}
\begin{tabular}{ccccc}
Class & Violation &\multicolumn{2}{l}{Inequality Representative for the Class}& Dimension \\
\hline 
11--40 & 1/4 &\multicolumn{2}{l}{$-r_{12}+r_{15}+r_{25}\leq 1$}&2\\
41--60 & 1/3 &\multicolumn{2}{l}{$+r_{15}+r_{25}+r_{35}-(r_{12}+r_{13}+r_{23})\leq 1$}&3\\
61--65 & 0.243 &\multicolumn{2}{l}{$+r_{12}+r_{13}+r_{14}+r_{15}-(r_{23}+r_{24}+r_{25}+r_{34}+r_{35}+r_{45})\leq 1$}&4\\
66--75 & 0.312 &\multicolumn{2}{l}{$+r_{12}+r_{14}+r_{15}+r_{23}+r_{34}+r_{35}-(r_{13}+r_{24}+r_{25}+r_{45})\leq 2$}&3\\
76--87 & $0.795$ &\multicolumn{2}{l}{$+r_{12}+r_{15}+r_{23}+r_{34}+r_{45}-(r_{13}+r_{14}+r_{24}+r_{25}+r_{35})\leq 2$}&2\\
88--92 & 0.344 &\multicolumn{2}{l}{$+2r_{12}+2r_{23}+2r_{24}+2r_{25}-(r_{13}+r_{14}+r_{15}+r_{34}+r_{35}+r_{45})\leq 3$}&4\\
93--152 & 0.688 &\multicolumn{2}{l}{$+r_{13}+r_{14}+2r_{24}+r_{34}+2r_{45}-(2r_{12}+2r_{25}+2r_{35})\leq 3$}&3\\
153--212 & 0.7306 &\multicolumn{2}{l}{$+2r_{12}+2r_{14}+2r_{15}+r_{23}+r_{35}-(2r_{13}+2r_{24}+r_{25}+2r_{45})\leq 3$}&2\\
213--242 & 0.855 &\multicolumn{2}{l}{$+2r_{13}+2r_{14}+2r_{23}+2r_{24}+3r_{35}+3r_{45}-(2r_{12}+4r_{15}+4r_{25}+r_{34})\leq 5$}&3\\
\end{tabular}
\end{ruledtabular}
\end{table*}
\endgroup

We now move to consider gluing along an edge. 

\begin{proposition}
Let $G_1$ and $G_2$ be graphs, $v_1,w_1 \in V(G_1)$ and $v_2,w_2 \in V(G_2)$ such that $e_i \defeq \enset{v_i,w_i} \in E(G_i)$. Writing
\[G \defeq G_1 +_{(v_1,w_1)=(v_2,w_2)} G_2,\]
we have
\begin{align*}
    C_{G} =& \setdef{r \in [0,1]^{E(G)}}{r|_{E(G_1)} \in C_{G_1}, r|_{E(G_2)} \in C_{G_2}}
    \\ \cong& \setdef{(r,s)}{r \in C_{G_1}, s \in C_{G_2}, r_{e_1} = s_{e_2}} 
    \\ \cong& (C_{G_1} \times [0,1]^{E(G_2)\setminus\enset{e_2}}) \cap ([0,1]^{E(G_2)\setminus\enset{e_1}} \times C_{G_2}) ,
\end{align*}
where for the last line we assume that $C_{G_1}$ is written with $e_1$ as its last coordinate and $C_{G_2}$ with $e_2$ as its first coordinate.
\end{proposition}
\begin{proof}
The proof is similar to that of \cref{prop:glue_vertex}, but now we are forced to make two identifications between elements of $\Lambda_1$ and of $\Lambda_2$ in $\Lambda_1 \sqcup \Lambda_2$.
When $\lambda_1$ and $\lambda_2$ are such that 
$\epsilon_{\lambda_1}(e_1) = \epsilon_{\lambda_2}(e_2)$,
\ie such that
\[ \lambda_1(v_1) = \lambda_1(w_1) \;\;\Leftrightarrow\;\; \lambda_2(v_2) = \lambda_2(w_2),\]
then this yields a well-defined vertex labelling of $G$ and the result follows.
\end{proof}

Note that the result is not quite as strong as \cref{prop:disjointunion,prop:glue_vertex}.
While the inequalities of $C_{G_1}$ plus those of $C_{G_2}$ form a complete set of inequalities for the classical polytope of the resulting graph $G_1 +_{(v_1,w_1)=(v_2,w_2)} G_2$, this is not necessarily a minimal set.

\begin{proposition}\label{prop:subgraph}
Let $G$ be a graph and $G'$ be a subgraph of $G$ on the same set of vertices,
\ie $V(G') = V(G)$ and $E(G') \subseteq E(G)$.
Then $C_{G}$ is a subpolytope of $C_{G'} \times [0,1]^{E(G)\setminus E(G')}$.
\end{proposition}
\begin{proof}
We need to show that the vertices of $C_{G}$ constitute a subset of the vertices of $C_{G'} \times [0,1]^{E(G)\setminus E(G')}$,
\ie that $\EqG \subseteq \EqGprime \times \enset{0,1}^{E(G)\setminus E(G')}$.
Given a classical edge labelling of $G$, \ie an edge labelling of the form 
$\epsilon_\lambda$ for some vertex labelling $\fdec{\lambda}{V(G)}{\Lambda}$,
we can regard $\lambda$ as a vertex labelling of $G'$ and
conclude that its equality labelling is simply the restriction of $\fdec{\epsilon_\lambda}{E(G)}{\enset{0,1}}$ to the subset $E(G')$ of its domain.
\end{proof}

In particular, $C_{K_n}$ is a subpolytope of $C_G$ for any event graph $G$ with $n$ vertices.

\section{Classical polytope facets and quantum violations for small graphs}\label{app:smallgraphs}

In this section, we study the facet-defining inequalities of some small graphs.
In particular, we analyze and classify the facet-defining inequalities for the classical polytopes $C_G$ corresponding to complete event graphs of 4 and 5 vertices ($G=K_4$ and $G=K_5$, respectively).
We also find quantum violations of these inequalities with pure states that are sampled from the set of quantum states.
For sampling we used the Python library \textrm{QuTip}~\cite{johansson2012qutip}.

Ref.~\cite{GalvaoB20} gave a complete characterization of the classical polytope of the graph $K_3 = C_3$,
the smallest graph with non-trivial inequalities,
together with a characterization of its maximal quantum violations, as well as applications.
More generally, Ref.~\cite{GalvaoB20} gave the complete set of inequalities for the classical polytope of the cycle graphs $C_n$, which take the very simple form in \cref{eq:ncycle}.
Here, we move to consider graphs with more than three edges and which are not cycles.

\psection{Facet-defining inequalities for small complete graphs}%
The facet-defining inequalities of the classical polytope of the graph $C_4$ (the $4$-cycle) are those of the form given by the CHSH inequality mentioned in the main text.
If we add one more edge to this graph, the corresponding polytope ends up being described by $3$-cycle inequalities alone.
Therefore, the first interesting graph yielding non-trivial and non-cycle inequalities is $K_4$, the complete graph of 4 vertices.
The classical polytope of this graph has facets defined by $3$- and $4$-cycle inequalities, together with facets defined by the new inequalities described in \cref{eq:k4nontrivial} in the main text,
\ie those of the form
\[ (r_{12}+r_{13}+r_{14})-(r_{23}+r_{34}+r_{24}) \le 1 . \]

This inequality has a structure that is present for all $K_n$ graphs, as will be discussed in \Cref{app:family}.
Since complete graphs have all possible edges, these are the graphs that impose the largest number of non-trivial constraints on edge assignments, as per \cref{prop:subgraph}.
Therefore, it is natural to look at those graphs first.

We addressed the complete characterization of the classical polytopes of complete graphs,
proceeding as far as the computational complexity of the problem allowed.
In particular, we found complete sets of facet-defining inequalities for $C_{K_5}$ and $C_{K_6}$.
The polytope $C_{K_5}$ has 52 vertices and 242 facets. These facets fall are divided into 9 symmetry classes. Representative inequalities from each of these classes are shown in \cref{tab: inequalities K5}.
The polytope $C_{K_6}$ has 203 vertices and requires 50,652 inequalities.
A list of inequalities and Python code used to obtain them can be found in Ref.~\cite{wagner2022github}.

\begin{figure*}[thb]
    \centering
    \includegraphics[width=.99\textwidth]{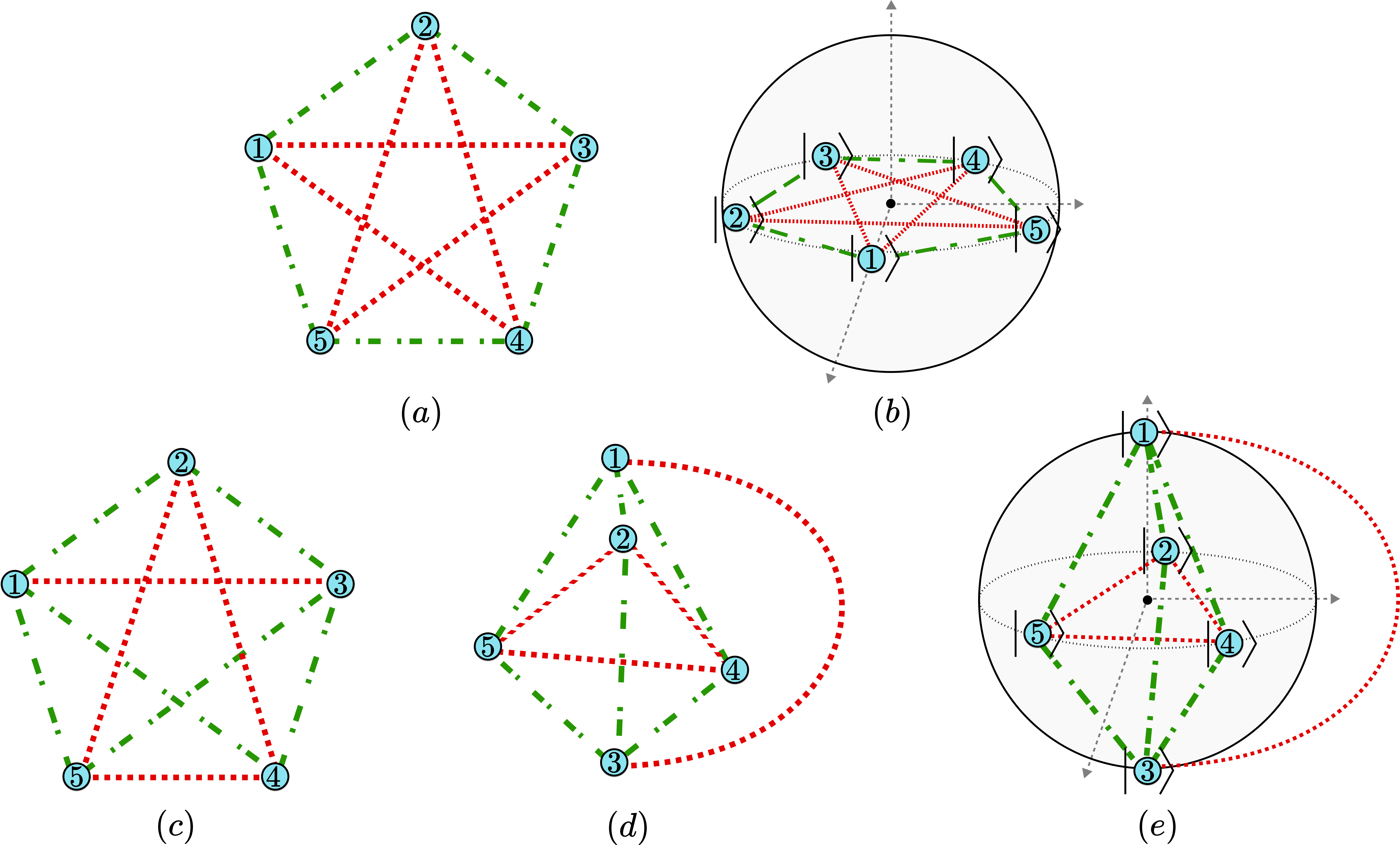}
    \caption{\textbf{Qubit states violating classicality inequalities.}
     (a) depicts the classicality inequality $r_{12}+r_{23}+r_{34}+r_{45}+r_{15}-r_{13}-r_{14}-r_{24}-r_{25}-r_{35}\leq 2$ with edges corresponding to positive terms in green (dash-dotted lines) and to negative terms in red (dashed-only lines). (b) shows a set of five pure states equally spaced over a great circle of the Bloch sphere, which violates this inequality attaining a value of $\sfrac{5\sqrt{5}}{4} > 2$. (c) depicts the classicality inequality $r_{12}+r_{14}+r_{15}+r_{23}+r_{34}+r_{35}-r_{13}-r_{24}-r_{25}-r_{45} \leq 2$ as in (a).
     (d) depicts the same inequality with the graph displayed in a different geometric configuration,  mirroring that of a set of states in the Bloch sphere that largely violates it.
     (e) represents that set of five pure states in the Bloch sphere: three states equally spaced around the equator plus the two eigenstates of the Pauli matrix $\sigma_z$; this set of states attains a value of $\sfrac{9}{4} > 2$ for the inequality.}
    \label{fig: qubit_violations}
\end{figure*}

\psection{Quantum violations}%
We looked for quantum violations of each inequality class of $C_{K_5}$
obtained by pure states in Hilbert spaces of dimensions 2, 3, and 4.
The violations found are included in \cref{tab: inequalities K5}.
The inequality in the third row is apparently not violated by either qubit or qutrit states.
The largest violation found among all the inequalities was $0.855$, for the inequality in the last row of the table.
The sets of quantum states yielding the violations found are presented in Ref.~\cite{wagner2022github}.

For some classes of inequalities,
we also found violations using pure qubit states that display interesting symmetries in the Bloch sphere.
We present those violations in \cref{fig: qubit_violations}.
For instance, consider the inequality in the fifth row of \cref{tab: inequalities K5}.
It can be violated with the quantum states
\begin{align}
    \vert \psi_k \rangle = \frac{1}{\sqrt{2}}\left(\vert 0 \rangle + e^{2\pi i k/5}\vert 1 \rangle\right)
\end{align}
with $k=0,\dots,4$.
This quantum realization attains a value of $\sfrac{5\sqrt{5}}{4}$ and hence a violation of $\sfrac{5\sqrt{5}}{4}-2\approx 0.79508$.
Another interesting violation with qubits is for the inequality in the fourth row of the table.
There, a maximal qubit violation is achieved by the states depicted in \cref{fig: qubit_violations}: choosing $\vert \psi_2 \rangle, \vert \psi_4 \rangle, \vert \psi_5 \rangle $ equally distributed on the equator of the Bloch sphere, \ie separated by angles of $\sfrac{2\pi}{3}$, implying that $r_{24}=r_{25}=r_{45}=\sfrac{1}{4}$, and choosing $\vert \psi_1 \rangle = \vert 0 \rangle ,\vert \psi_3 \rangle = \vert 1 \rangle$, implying that $r_{13}=0$ and all remaining overlaps are equal to $\sfrac{1}{2}$.
This set of vectors attains the value $\sfrac{6}{2}-\sfrac{3}{4}=\sfrac{9}{4}$ and hence a violation of $\sfrac{9}{4}-2=\sfrac{1}{4}$.
These symmetrically arranged qubit states are also the states used in the construction of the elegant joint measurement of Ref.~\cite{gisin2019entanglement}.
However, we could find a higher violation of the same inequality using qutrits, as shown in the table.

We will see in \cref{app:family} that the inequality in \cref{eq:k4nontrivial} generalizes to an infinite family of inequalities for the polytope of $K_n$.
The quantum violation found for this non-cycle $K_4$ inequality used the following four qutrit states:
\begin{align*}
    \ket{\psi_1} &= \ket{0} \\
    \ket{\psi_2} &= \sqrt{\frac{5}{9}}\ket{0} + \sqrt{\frac{4}{9}}\ket{1} \\
    \ket{\psi_3} &= \sqrt{\frac{5}{9}}\ket{0} - \sqrt{\frac{1}{9}}\ket{1} + i\sqrt{\frac{1}{3}}\ket{2} \\
    \ket{\psi_4} &= \sqrt{\frac{5}{9}}\ket{0} - \sqrt{\frac{1}{9}}\ket{1} - i\sqrt{\frac{1}{3}}\ket{2}\\
\end{align*}
This set of states attains a value of $\sfrac{4}{3}$ and hence violation of $\sfrac{4}{3}-1=\sfrac{1}{3}$.
This corresponds to the second class of inequalities of $C_{K_5}$ in \cref{tab: inequalities K5}.

We remark once more that the above violations are \stress{not} necessarily optimal.
They were not found using \eg techniques of semidefinite programming over the quantum set.
We found this landscape of violations by sampling quantum states and calculating the value of the left-hand side of the inequality, which is suboptimal.
An important  remark is that the quantum violation for the $3$-cycle inequality class (first row in \cref{tab: inequalities K5})
is \stress{provably maximal}, as shown in Ref.~\cite{GalvaoB20}.

\section{Infinite family of \\ classical polytope facets}\label{app:family}

\Cref{eq:k4nontrivial} in the main text shows a facet-defining inequality of the polytope $C_{K_4}$
that is not of the previously known form of inequalities derived from cycles in Ref.~\cite{GalvaoB20}
(which where enough, incidentally, to characterise the classical polytope of the graph $K_3 = C_3$).
In this section, we generalize it to an infinite family of new classicality inequalities.
More concretely, we present a construction of a facet-defining inequality of the classical polytope $C_{K_n}$ for any $n \geq 2$.
Moreover, each inequality on this family cannot be obtained from combining prior members of the family.
For $n = 4$, this recovers the just-mentioned inequality from \cref{eq:k4nontrivial}, while for $n=3$ it naturally reduces to the 3-cycle inequality.

Fix a natural number $n \geq 2$.
Write $V_n = \enset{1, \ldots, n}$ for the vertices of $K_n$, and let
$E_n$ denote the set of edges of $K_n$, \ie all two-element subsets of $V_n$.
Consider a partition of $E_n$ into the subsets $G_n, R_n \subseteq E_n$ given as  
\begin{align*}
G_n &\defeq \setdef{\enset{1,i}}{i=2,\ldots,n}
\\
R_n &\defeq E_{n} \setminus G_n.
\end{align*}
The edges in $R_n$ determine a complete subgraph of $K_n$ with one fewer vertex, \ie a subgraph isomorphic to $K_{n-1}$.
In turn, the edges in $G_n$ form a subgraph isomorphic to $K_{1,n-1}$, a star graph with $n$ vertices.  
We use this specific partition of $E_n$ to define a generalized version of the inequality from \cref{eq:k4nontrivial}:  
\begin{equation}\label{eq: new K4 generalized}
    h_n(r) \defeq \sum_{e\in G_n}r_e - \sum_{e\in R_n} r_e \leq 1.
\end{equation}

\begin{figure}[thb]
    \centering
    \includegraphics[width=0.3 \textwidth]{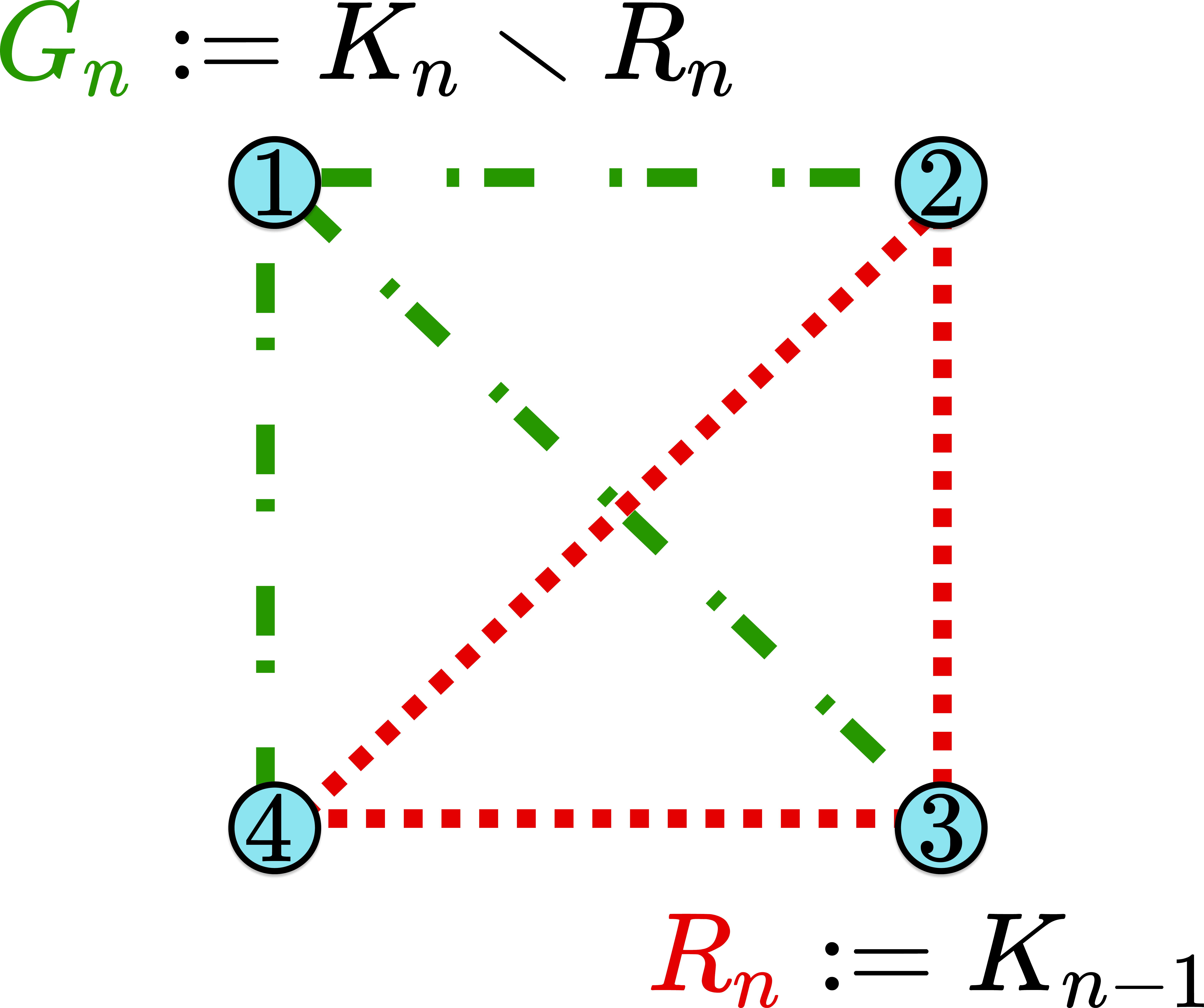}
    \caption{\textbf{Depiction of the sets $R_n$ and $G_n$ for a given complete graph $K_n$.} The set $R_n$ is always a complete subgraph (isomorphic to)  $K_{n-1}$ of $K_n$. Here we considered $n=4$ as an example.}
    \label{fig: R_and_G_figure_cut_Kn}
\end{figure}

We first show that this is indeed a classicality inequality for the complete graph $K_n$.
\begin{proposition}
    For any $n \geq 2$, the classical polytope $C_{K_n}$ of the complete event graph $K_n$ is contained in the half-space defined by the inequality $h_n$ from \cref{eq: new K4 generalized}, \ie all classical edge weightings of $K_n$ satisfy the inequality.
\end{proposition}
\begin{proof}
It suffices to check that the inequality is satisfied by any vertex of the polytope $C_{K_n}$.
Recall that the vertices of this polytope correspond to classical edge $\enset{0,1}$-labellings of the graph $K_n$,
that is, those realisable as the equality labelling of some vertex labelling.

So, let $\fdec{\lambda}{V_n}{\Lambda}$ be any vertex labelling and $r \in [0,1]^{E_n}$ be the vertex of the classical polytope corresponding to its equality edge labelling.
That is, for all $e=\enset{i,j} \in E_n$,
\[r_{ij} = \epsilon_\lambda(\enset{i,j})=\begin{cases}1 & \text{ if $\lambda(i)=\lambda(j)$} \\ 0 & \text{ if $\lambda(i)\neq\lambda(j)$}\end{cases}.\]
Consider the set of vertices in $\enset{2, \ldots, n}$ that are labelled the same as vertex $1$,
\begin{align*}
S_\lambda &= \setdef{i \in \enset{2, \ldots, n}}{\lambda(i)=\lambda(1)} 
\\&= \setdef{i \in \enset{2, \ldots, n}}{r_{1 i} = 1}
\end{align*}
By construction, an edge in $G_n$, which is of the form $\enset{1,i}$, is labelled $1$ or $0$ depending on whether $i$ is in $S_\lambda$ or not.
Moreover, by transitivity of equality, if $i,j \in S_\lambda$ then $\lambda(i)=\lambda(j)$, meaning that
the edge $\enset{i,j}$ is also labelled $1$.
Writing $k \defeq |S_\lambda|$, one can therefore bound the left-hand side of \cref{eq: new K4 generalized}:
\begin{align*}
\sum_{e\in G_n}r_e - \sum_{e\in R_n} r_e 
&=
\sum_{i \in S_\lambda}r_{1i}
- \sum_{e\in R_n} r_e 
\\
&= k - \sum_{e\in R_n} r_e 
\\
&\leq k - \sum_{i,j \in S_\lambda} r_{ij}
\\
&= k - \binom{k}{2}
\\&=
1-\binom{k-1}{2} 
\\&\leq 1
\end{align*}
where for the corner case $k=0$ this still holds putting $\binom{-1}{2}=1$
\end{proof}

We now state the central result of this section.

\begin{theorem}
The inequality $h_n$ from  \cref{eq: new K4 generalized} defines a facet of the classical polytope $C_{K_n}$ of the complete event graph $K_n$ for any $n \geq 2$.
\end{theorem}
\begin{proof}
We establish this result by finding the set of vertices of the polytope $C_{K_n}$ that belongs to -- and therefore determines -- this facet.
In fact, it suffices to find a set of points $F$ in the space (of edge weightings) such that:
(i) all the points in $F$ belong to the polytope $C_{K_n}$,
(ii) all the points in $F$ saturate the inequality, \ie belong to the hyperplane determined by it,
(iii)
the set $F$ is affinely independent,
and (iv) $F$ contains as many points as the dimension $D$ of the polytope, so that it generates an affine subspace of dimension $D-1$.
In our proof, the chosen points are moreover vertices of the polytope, as they are edge $\enset{0,1}$-labellings.

We construct a set $F$ of polytope vertices.
This consists of two kinds of edge labellings:
those that assign $1$ to exactly one edge of $G_n$ (and $0$ to all other edges of $E_n$) and those that assign $1$ precisely to a triangle consisting of two edges from $G_n$ and another from $R_n$.
More formally, we define a family of edge $\enset{0,1}$-labellings indexed by subsets of size $1$ or $2$
of the vertex set $\enset{2,\ldots,n}$, as follows:
for each $i = 2, \ldots, n$,
define the edge $\enset{0,1}$-labelling $r^{(i)}$
with
$r^{(i)}_{1i} = 1$
and $r^{(i)}_e=0$ for all other edges $e$;
for each pair
$i,j = 2, \ldots, n$ with $i \neq j$
define the edge $\enset{0,1}$-labelling $r^{(i,j)}$
with $r^{(i,j)}_{1i} = r^{(i,j)}_{1j} = r^{(i,j)}_{ij} = 1$
and $r^{(i,j)}_e=0$ for all other edges $e$.
The set $F$ is then given by
\[
F \defeq \setdef{r^{(i)}}{i = 2, \ldots, n} 
\cup \setdef{r^{(i,j)}}{i,j = 2, \ldots ., n, i \neq j}.
\]
\Cref{fig:principle_argument_newk4_infinite} depicts the construction of the set $F$ for the case of $n=5$.
We now check conditions (i)--(iv) to establish the desired result.

\begin{figure}[thb]
    \centering
    \includegraphics[width=\columnwidth]{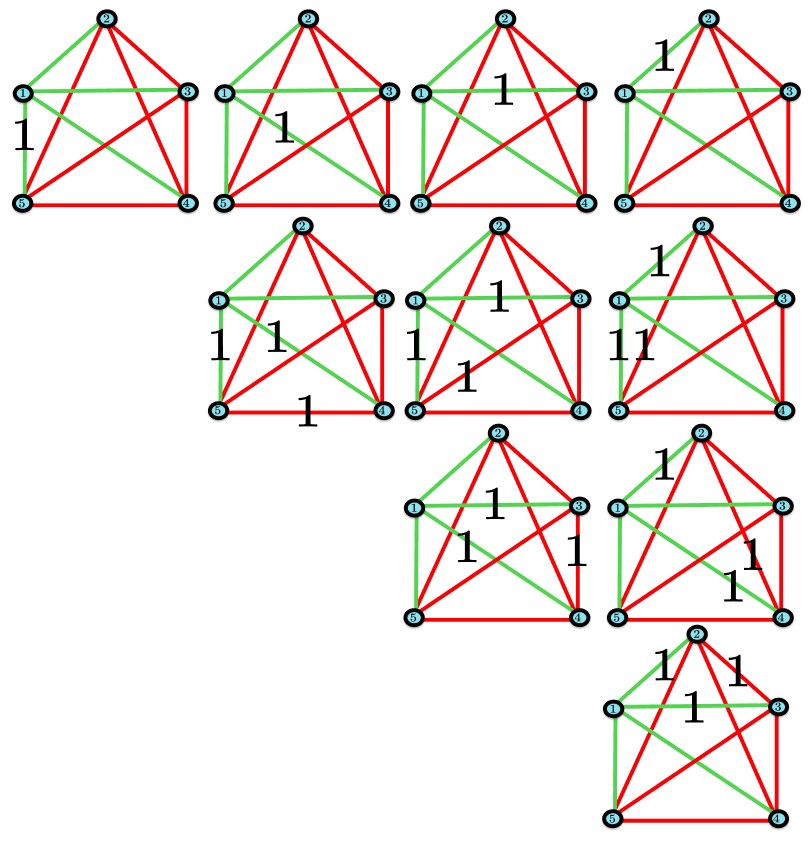}
    \caption{\textbf{The construction of the set $F$ for the $K_5$ graph.} Each edge is labelled $1$ where explicitly noted, otherwise it is labelled $0$ (to keep the figures easy to read). The first row shows the four labellings of the form $r^{(i)}$ with only one edge labelled $1$ from $G_5$.
    The remaining rows show the labellings of the form $r^{(i,j)}$, which assign label $1$ to exactly one triangle consisting of two edges from $G_5$ and the connecting edge from $R_5$.}
    \label{fig:principle_argument_newk4_infinite}
\end{figure}

For condition (i), we use \cref{prop:charclassicalvertices} to show that all the edge labellings in the set $S$ are classical and thus vertices of the polytope $C_{K_n}$.
Indeed, no cycle can have exactly one edge with label $0$. In the case of the labellings of the form $r^{(i)}$, this is immediate as there is only one edge not labelled $0$.
For the labellings of the form $r^{(i,j)}$, no triangle (\ie subgraph isomorphic to $C_3$) has exactly one edge labelled $0$: if one chooses two edges labelled $1$ then the remaining edge that completes the 3-cycle also has label $1$. Moreover, any larger cycle can have at most two edges labelled $1$.
Alternatively, we can show this by constructing an underlying vertex labelling:
for $r^{(i)}$ pick $\fdec{\lambda}{V_n}{\Lambda}$
with $\lambda(1)=\lambda(i)$ and all other vertex labelled differently;
for $r^{(i,j)}$ pick $\lambda$ with $\lambda(1)=\lambda(i)=\lambda(j)$
and the other vertices labelled differently.

Condition (ii) is directly checked: for each $i = 2, \ldots, n$ we have
\[
\sum_{e\in G_n}r^{(i)}_e - \sum_{e\in R_n} r^{(i)}_e = r^{(i)}_{1i} - 0 = 1 - 0 = 1 ,
\]
and for each pair $i, j = 2, \ldots, n$ with $i \neq j$,
\[
\sum_{e\in G_n}r^{(i,j)}_e - \sum_{e\in R_n} r^{(i,j)}_e = r^{(i,j)}_{1i} + r^{(i,j)}_{1j} - r^{(i,j)}_{ij} = 2 - 1 = 1. 
\]

For condition (iii), 
affine independence can be verified by inspecting the matrix whose columns are the vectors corresponding to the edge-labellings in $F$. Ordering the components of each vector (corresponding to the edges of $K_n$) in lexicographic order and listing $r^{(i)}$ followed by $r^{(i,j)}$ also in that order, the matrix is arranged to be triangular with diagonal entries all equal to 1, hence its determinant is equal to $1$, implying linear independence of the vectors.

Finally, for condition (iv), 
as all these labellings are distinct, one can count the number of elements of $S$ from the way they were constructed:
\[|F| = \binom{n-1}{1} + \binom{n-1}{2} = \binom{n}{2} = \frac{n(n-1)}{2}.\]
We conclude that it is the same as the dimension of the ambient space (of edge labellings) where the polytope lives, and thus also of the polytope itself.
\end{proof}

\section{Event graphs and \\ Kochen--Specker contextuality}\label{app:contextuality}

In this section, we establish a formal connection between our framework and (Kochen--Specker) contextuality.
The central result (\cref{theorem: STAB = C}) shows how our event graph formalism recovers all noncontextuality inequalities obtainable from the Cabello--Severini--Winter (CSW) exclusivity graph approach~\cite{cabello2014graph}.

To achieve this, we encode a contextuality setup, represented in CSW by an exclusivity graph $H$,
by imposing exclusivity constraints on a related event graph $H_\star$.
This process amounts to taking a cross-section yielding a subpolytope of the classical polytope $C_{H_\star}$.
We show that the resulting facet inequalities bound noncontextual models for $H$.

In fact, we prove something \stress{stronger}.
We describe an explicit isomorphism between
the noncontextual polytope associated by CSW to the exclusivity graph $H$
and
this cross-section subpolytope of the classical polytope $C_{H_\star}$ associated by our approach to the event graph $H_\star$.
In particular, these polytopes have the same non-trivial facet-defining inequalities.
These are obtainable from the inequalities that define the full (unconstrained) classical polytope of the event graph $H_\star$ by setting some coefficients to zero.
\Cref{theorem: STAB = C} thus establishes a tight correspondence between our event graph approach and a broad, well-established framework for contextuality.

In what follows, we introduce the relevant definitions regarding the exclusivity graph approach, the associated event graphs, and the constraints to be imposed on them, before proving the new results. 

\psection{The exclusivity graph approach}%
In the CSW framework from Ref. \cite{cabello2014graph}, contextuality scenarios are described by so-called exclusivity graphs. Hence this formalism is also known as the exclusivity graph approach; see also~\cite[Chapter 3]{amaral2018graph} for a recent and comprehensive discussion.

The vertices of an exclusivity graph $H$ represent measurement events, and its edges indicate
exclusivity between events, where two events are exclusive that can be distinguished by a measurement procedure.

Even though the CSW framework is theory-independent, it is helpful for clarity of exposition to consider its instantiation in quantum theory, in order to better convey the underlying intuitions.
In quantum theory, measurement events are represented by projectors (PVM elements) on a Hilbert space, or equivalently, by closed subspaces of the Hilbert space.
Exclusivity is captured by orthogonality, which characterizes when two projectors may appear as elements of the same PVM, \ie events from \stress{the same} measurement procedure.
Given a set of projectors $\{\Pi_v\}_{v \in V}$ on a fixed Hilbert space,
the corresponding contextuality scenario is thus described by 
its orthogonality graph.
This graph has set of vertices $V$ and has an edge $\{u,v\}$ if and only if the projectors $\Pi_u$ and $\Pi_v$ are orthogonal to each other, \ie when $\Pi_u\Pi_v=0$.

In this approach, a non-negative vertex weighting $\fdec{\gamma}{V(H)}{\RR_{\geq 0}}$ on the exclusivity graph $H$
determines a noncontextuality inequality
on the probabilities $P(v)$ of measurement events $v \in V(H)$:
\[\sum_{v \in V(H)} \gamma(v) P(v)   \;\leq\; \alpha(H,\gamma),\]
where $\alpha(H,\gamma)$ is the independence number of the vertex-weighted graph.
In the quantum case, this yields a noncontextuality condition on the
statistics predicted by a given quantum state $\psi$:
\[\sum_{v \in V(H)} \gamma(v) \langle \psi \vert \Pi_v \vert \psi \rangle   \;\leq\; \alpha(H,\gamma) .\]

Such noncontextuality inequalities above determine the polytope of noncontextual behaviours for  any exclusivity graph $H$.
This polytope, known as the stable set polytope of $H$, $\STAB(H)$,
is most readily defined by its V-representation, which we now present, following \cite[Chapter 3]{amaral2018graph}.

\begin{definition}\label{def: stable set}
Let $H$ be a graph. A subset $S \subseteq V(H)$ of vertices is called a \stress{stable set} if no two vertices of $S$ are adjacent in $H$,
\ie for all $v, w \in S$, $\{v,w\} \not\in E(H)$.
Write $\mathcal{S}(H)$ for the set of stable sets of $H$.
\end{definition}

To any subset of vertices $W\subseteq V(H)$ corresponds
its characteristic map, the vertex $\{0,1\}$-labelling $\fdec{\chiW}{V(H)}{\enset{0,1}}$ given by: 
\[
    \chiW(v) \defeq \begin{cases}
    1 & \text{ if $v \in W$,}\\
    0 & \text{ if $v \notin W$.}
    \end{cases}
\]
Through the inclusion $\enset{0,1} \subseteq [0,1]$, one regards a vertex $\{0,1\}$-labelling (equivalently, a subset of vertices) as a point in $[0,1]^{V(H)} \subseteq \RR^{V(H)}$.
Those arising from stable sets $S \in \mathcal{S}(H)$ correspond to the deterministic noncontextual models, which determine the whole convex set of noncontextual behaviours.

\begin{definition}
\label{def: STAB}
The \stress{stable set polytope} of a graph $H$, denoted $\STAB(H)$, 
is the convex hull of the points $\chiS \in [0,1]^{V(H)}$ with $S$ ranging over all stable sets of $H$,
\[
\STAB(H) \defeq \textrm{ConvHull}\setdef{\chiS}{S \in \mathcal{S}(H)}.
\]
\end{definition}

To get the intuition underlying this description, one may think of a vertex $\{0,1\}$-labelling $\fdec{\chi_W}{V(H)}{\enset{0,1}}$ as a deterministic assignment of truth values to all measurement events (vertices of the exclusivity graph). In this interpretation, the subset of vertices $W \subseteq V(H)$ is the set of measurement events that are assigned \textit{true}.
The stability condition indicates that no two adjacent vertices of the exclusivity graph $H$ are labelled with $1$, that is, two exclusive measurement events cannot be simultaneously true.
This captures the exclusivity condition at the deterministic level, thus yielding the deterministic noncontextual models.

\psection{From exclusivity graphs to constrained event graphs}%
We relate this approach to our framework by constructing a new (event) graph $H_\star$ from any (exclusivity) graph $H$. This is obtained by adding a new vertex $\psi$ with an edge connecting it to all the vertices of $H$.
See \cref{fig:orthogonality_to_event} for an instance of this construction for the KCBS scenario and \cref{fig:ProofAppD} for a more generic description.
The construction is formally described in \cref{def: G out of G prime} below.

The relevance of the new vertex $\psi$ is well known;
it is usually called the `handle' and it appears in the literature on the graph approaches~\cite{amaral2018graph,baldijao2020classical,vandre2022quantum}.
Its name comes from the geometric arrangement of the vectors providing the maximal quantum violation of the KCBS inequality of \cref{eq: KCBS event graph}:
the quantum state resembles the handle of an umbrella made of the vectors that describe measurement events.

\begin{definition}\label{def: G out of G prime}
Let $H$ be a graph. Define a new graph $H_\star$ by
\begin{align*}
    V(H_\star) &\defeq V(H) \sqcup \enset{\psi}
    \\
    E(H_\star) &\defeq E(H) \cup \setdef{\enset{\psi,v}}{v \in V(H)}.
\end{align*}
Moreover, define $C_{H_\star}^0$ to consist of the classical edge weightings of $H_\star$ that assign value $0$ to all edges in $H$,
\[
C_{H_\star}^0 \defeq \setdef{r \in C_{H_\star}}{\Forall{e\in E(H)} r_e=0} .
\]
\end{definition}

The set $C_{H_\star}^0$ is, by construction, a cross-section of the classical polytope $C_{H_\star}$ of the event graph $H_\star$,
being its intersection with the $|V(H)|$-dimensional subspace defined by the equations $\bigwedge_{e \in E(H)}r_e=0$.
Moreover, it is a subpolytope of $C_{H_\star}$, \ie the convex hull of a subset of its vertices. These vertices are the classical edge $\enset{0,1}$-labellings that assign label $0$ to edges in $H$.
In terms of the underlying vertex labellings (from which classical edge labellings arise as equality labellings), the requirement is that any two vertices adjacent in $H$ must be labelled differently.

\begin{figure}[t]
    \centering
    \includegraphics[width=\columnwidth]{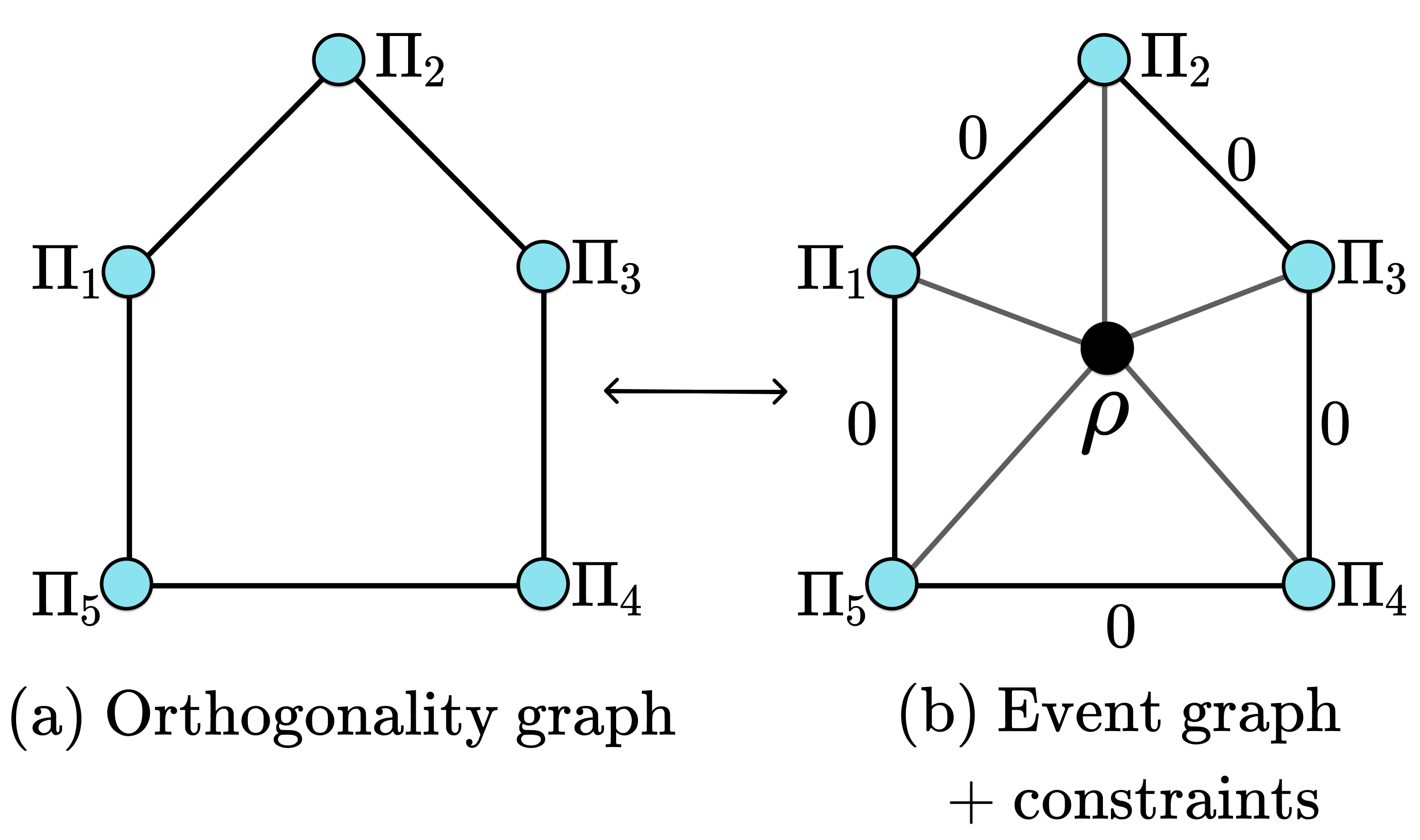}
    \caption{\textbf{Equivalence described by \cref{theorem: STAB = C} linking contextuality à la CSW to event graphs.}
    The behaviours on an exclusivity graph are in bijective correspondence with edge weightings (overlap assignments) in the related event graph subject to constraints.
    In particular, the \stress{noncontextual} behaviours for the exclusivity graph correspond bijectively to the \stress{classical} edge weightings in the event graph with constraints.}
    \label{fig:orthogonality_to_event}
\end{figure}

\psection{Recovering the noncontextual polytope}%
The edge set of the graph $H_\star$ can be partitioned into two sets: the edges already present in $H$ and the new edges of the form $\{\psi,v\}$ for $v \in V(H)$.
The latter are in one-to-one correspondence with vertices of $H$. So, there is a bijection $E(H_\star) \cong E(H) \sqcup V(H)$.

When considering the polytope $C_{H_\star} \subseteq [0,1]^{E(H_\star)}$
we adopt the convention of ordering the coordinates with the edges already in $H$ listed first, so that
\[\RR^{E(H_\star)} \cong \RR^{E(H) \sqcup V(H)} \cong \RR^{E(H)} \times \RR^{V(H)}.\]
The subpolytope $C_{H_\star}^0$ is thus written as the set of points
of $C_{H_\star}$ of the form $(\mathbf{0}_{H},r)$ where 
$\mathbf{0}_{H}$ is the zero vector in $\RR^{E(H)}$ (corresponding to the edges inherited from $H$)
and $r$ is a weighting of the remaining (new) edges.
In particular, the vertices of $C_{H_\star}^0$ are precisely the classical $\{0,1\}$-labellings of $H_\star$
that assign the label $0$ to all the edges in $H$.

We can now prove our main result, showing that $C_{H_\star}$ is indeed (isomorphic to) the polytope of noncontextual behaviours for $H$.

\begin{theorem}\label{theorem: STAB = C}
For any (exclusivity) graph $H$, there is an isomorphism of polytopes
\[C_{H_\star}^0 \; \cong \; \STAB(H)\]
between the stable set polytope (of noncontextual models) of $H$
and the subpolytope of the classical polytope of event graph $H_\star$ constrained by the exclusivity conditions.
More explicitly, this is given by the identification
\[C_{H_\star}^0 \; = \; \{\mathbf{0}_H\} \times \STAB(H)\]
where $\mathbf{0}_H$ is the zero vector in $\mathbb{R}^{E(H)}$.
\end{theorem}
\begin{proof}
To establish the result, we consider the vertices of these polytopes.
Per the above discussion, we have $E(H_\star) \cong E(H) \sqcup V(H)$.
Consequently, there is a bijection between
vertex $\{0,1\}$-labellings of $H$ (equivalently, subsets of $V(H)$), on the one hand,
and edge $\{0,1\}$-labellings of $H_\star$ that assign label $0$ to all the edges in $E(H)$, on the other.
Explicitly, to each subset of vertices $W \subseteq V(H)$ corresponds the edge-labelling of $H_\star$
\[\fdec{[\mathbf{0}_H,\chi_{W}]}{E(H_\star) \cong E(H) \sqcup V(H)}{\{0,1\}},\]
as depicted in \cref{fig:ProofAppD}.

We show that this bijection restricts to a bijection between the \stress{classical} assignments in both cases.
Concretely, a subset of vertices $S \subseteq V(H)$ is \stress{stable}, hence (its characteristic map $\fdec{\chi_{S}}{V(H)}{\{0,1\}}$ is) a vertex of the polytope $\STAB(H)$,
if and only if the corresponding edge labelling $[\mathbf{0}_H,\chi_{S}]$ of $H_\star$ is classical, hence a vertex of the polytope $C_{H_\star}$ and thus of $C_{H_\star}^0$. 

We establish the two directions of this equivalence simultaneously, recalling the characterisation of classical edge labellings from \cref{prop:charclassicalvertices}.
Consider $H_\star$ with edge labelling $[\mathbf{0}_H,\chi_{S}]$.
The labelling fails to be classical if and only if there is an edge with label $0$ between two vertices linked by a path consisting of edges with label $1$.
Since all the edges between vertices in $H$ have label $0$,
the only way to build such a path of $1$-labelled edges is via the handle $\psi$: \eg $\{u,\psi\},\, \{\psi, v\}$ where both $u$ and $v$ must belong to $S$.
So, two vertices $u$ and $v$ of $H_\star$ are linked by a $1$-labelled path if and only if they both belong to $S \cup \{\psi\}$.
Therefore, the labelling is classical if and only if there is no edge with label $0$ between vertices in this set $S \cup \{\psi\}$.
To further simplify this condition, note that edges between $\psi$ and a vertex from $S$ have label $1$ by construction of the second component of $[\mathbf{0}_H,\chi_{S}]$, while from the first component, all edges between vertices in $H$ have label $0$.
The classicality condition is thus equivalent to there being no edges in $H$ between vertices in $S$, which is precisely to say that $S$ is stable.
\end{proof}

\begin{figure}[t]
    \centering
    \includegraphics[width=0.9 \columnwidth]{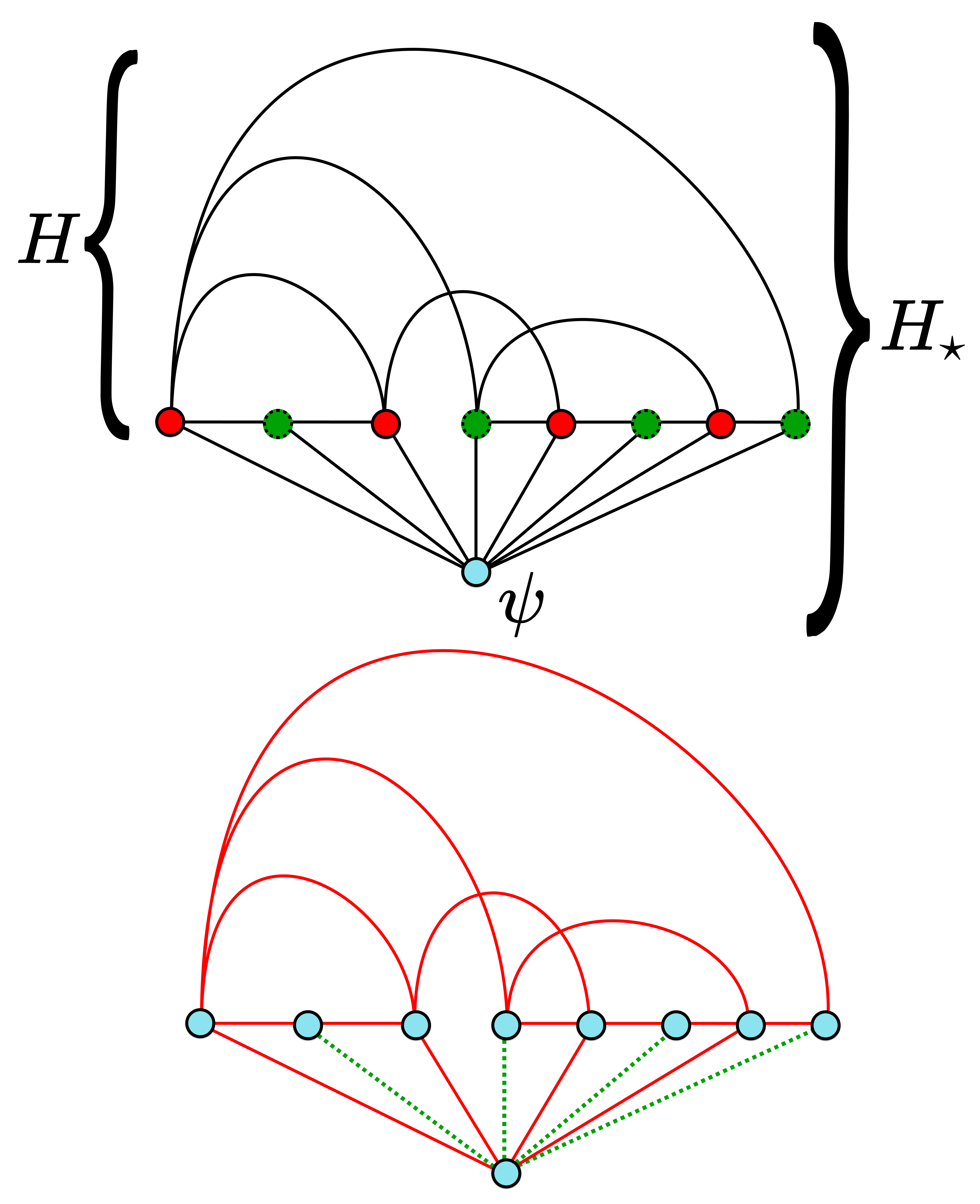}
    \caption{\textbf{Translation between vertex labellings of $H$ (that are characteristic maps of stable sets, hence vertices of $\STAB(H)$)
    and constrained edge labellings of $H_\star$ (that are classical, hence vertices of  $C_{H_\star}^0)$.)}
    The top figure depicts a graph $H$, standing for a generic exclusivity graph, and its extension $H_\star$  by adjoining the handle $\psi$ and new edges $\{\psi,v\}$ for all $v \in V(H)$.
    The vertices of $H$ that are shown in green (dashed) form a stable set $S \in \mathcal{S}(H)$.
    Its characteristic map $\fdec{\chi_S}{V(H)}{\{0,1\}}$ assigns $1$ to the green (dashed) vertices and $0$ to the red (solid) vertices of $H$.
    The bottom figure shows how such a vertex $\{0,1\}$-labelling is translated to an edge $\{0,1\}$-labelling of $H_\star$ assigning $0$ to all the edges of $H$ (and vice-versa) as described in the proof of \cref{theorem: STAB = C}. Green (dashed) edges are labelled $1$ and red (solid) edges are labelled $0$,
    in accordance with the vertex labellings from $\chi_S$, complemented by the labels induced by exclusivity constraints, $\mathbf{0}_{H}$, as described in the text.
    $S$ being stable is equivalent to the resulting edge labelling of $H_\star$ being classical.}
    \label{fig:ProofAppD}
\end{figure}

\psection{Recovering all noncontextuality inequalities}
We established \cref{theorem: STAB = C} in terms of the vertices of the polytopes, \ie by working with their V-representations.
We now consider the relationship between their H-representations, \ie their facet-defining inequalities.\footnote{H-representation is standard terminology referring to the description of a polytope as an intersection of half-spaces, \ie in terms of (facet-defining) inequalities. The `H' in `H-representation' is not to be confused with the symbol `$H$' that we use to denote an exclusivity graph.}

Of course, there is also a bijection between the facets of $\STAB(H)$ and those of $C_{H_\star}^0$.
Given the particularly simple description of the isomorphism, whereby $C_{H_\star}^0$ is written as a product of polytopes,
we can write this correspondence explicitly.
It turns out that the facet-defining inequalities of the subpolytope $C_{H_\star}^0$
are precisely the same as the facet-defining inequalities of the stable polytope of $H$.
Moreover, these can be obtained from the inequalities defining the (unconstrained) polytope $C_{H_\star}$ of the event graph $H_\star$ by setting some coefficients to zero.
We thus recover the full set of noncontextuality inequalities from our event graph formalism.

To see this, recall that if $P$ and $Q$ are two convex polytopes with H-representations
$P = \setdef{x}{A_1 \, x \preccurlyeq b_1}$ 
and
$Q = \setdef{y}{A_2 \, y \preccurlyeq b_2}$
then their product has H-representation
\[
P \times Q = \setdef{(x,y)}{A_1 \, x \preccurlyeq b_1 \text{ and } A_2 \, y \preccurlyeq b_2}.
\]
Here, the notation $A \, z \preccurlyeq b$ describes a set of linear inequalities on $z$ in matrix form, with the symbol $\preccurlyeq$ standing for component-wise inequality $\leq$ between real numbers.

Applying this to
\begin{align*}
C_{H_\star}^0 &= \{\mathbf{0}_{H}\} \times \STAB(H)\\
&=
\setdef{(x,y)}{x \in \{\mathbf{0}_{H}\}, y \in \STAB(H)}.
\end{align*}
we obtain that the 
H-representation of $C_{H_\star}^0$ is the conjunction
of the H-representations of $\{\mathbf{0}_{H}\}$ and of $\STAB(H)$.
The former consists simply of the equations $r_e = 0$ for each $e \in E(H)$,
zeroing out the first components, which corresponds to the weights of edges already in $H$.
Thus the non-trivial inequalities bounding $C_{H_\star}^0$ are thus the same as the inequalities bounding $\STAB(H)$.

Since $C_{H_\star}^0$ is obtained from  $C_{H_\star}$ by intersecting with the subspace that zeroes the components corresponding to edges in $E(H)$, a complete set of inequalities for $C_{H_\star}^0$ can be obtained from the facet-defining inequalities of $C_{H_\star}$
by disregarding those components, \ie setting the corresponding coefficients to zero.

This process is illustrated by the derivation of the KCBS inequality presented in the main text.
There, the exclusivity graph is the 5-cycle, with neighbouring vertices representing orthogonal projectors.
The graph $H_\star$ is then the 6-vertex wheel graph $W_6$ of \cref{fig:all_graphs_together}--(e).
As shown in the main text,
the KCBS noncontextuality inequality $\sum_a\gamma_a\vert\langle \psi \vert a \rangle \vert^2 \leq \alpha(H,\gamma)$
arises as a $C_{H_\star}^0$ inequality, being obtained from a classicality inequality for the event graph $W_6$ (a facet-defining inequality of $C_{H_\star}$) by setting to zero the coefficients relating to edges already in $H$.

\section{Event graphs and \\ preparation contextuality}\label{app:preparationcontextuality}

In this section, we relate our approach to Spekkens's notion of preparation contextuality.
This may be understood as providing a \stress{theory-independent} perspective
on the use of our formalism to witness quantum coherence.
There, the vertices of event graphs were interpreted as quantum states and the edges as two-state overlaps.
A similar treatment can be carried out
for a certain class of operational theories which support a notion of confusability,
with vertices interpreted as (abstract) preparation procedures.

\psection{Operational probabilistic theories}
Spekkens's notion of generalized contextuality is associated to operational probabilistic theories~\cite{dariano2017quantum,schmid2020structure,kunjwal2019beyondcabello}.
The description of an operational theory starts with a set of basic (operational) physical processes: in the simplest scenarios, one considers preparations and measurements.
One considers experiments consisting of a preparation $P$ followed by a measurement $M$ that returns an outcome $k$.
A probability rule associates a probability $p(k \mid M, P)$
of obtaining outcome $k$ when performing measurement $M$ after the preparation $P$.
More precisely, it associates a probability distribution over outcomes $k$ to each choice of preparation $P$ and measurement $M$.
For a dichotomic measurement $M$, \ie one with only two possible outcomes $0$ and $1$, we simplify notation and write $p(M \mid P)$ for $p(1 \mid M, P)$.
A crucial -- if sometimes overlooked -- aspect is that the full set of procedures includes also classical probabilistic mixtures (\ie convex combinations) of basic procedures, with the probability rule extended accordingly (\ie linearly).

Given an operational theory, one defines an equivalence relation identifying indistinguishable procedures.
Following Ref.~\cite{spekkens2005contextuality}, two preparation procedures are \textit{operationally equivalent}, written $P \simeq P'$, 
if and only if for all measurements $M$ and possible outcomes $k$,
\[
     p(k|M,P) = p(k|M,P') .
\]
A similar definition applies to measurement procedures, but this will not be needed in what follows.

When one treats quantum theory as an operational theory, quantum states $\vert \phi \rangle$ correspond to equivalence classes of operational procedures. For instance, a state $\ket{0}$ may represent preparing a ground state of a nitrogen atom, or preparing the horizontal polarization in photonic qubits. We relax this terminology and refer to `the preparation $P$ associated with a state $\ket{\phi}$', even though strictly speaking $P$ is only an instance of an equivalence class of procedures. 
Such relaxation is safe for our purposes.
In effect, it corresponds to treating pure quantum states as the basic procedures.
The interesting operational equivalences relevant for preparation contextuality
go beyond these, holding between classical mixtures of basic procedures.
For example, in quantum theory, the preparation procedure corresponding of an equal mixture of pure qubit states $\ket{0}$ and $\ket{1}$
is operationally equivalent to that corresponding to an equal mixture of states $\ket{+}$ and $\ket{-}$.
Indeed, both these classical mixtures define the same qubit mixed state, the totally mixed state.

\psection{LSSS operational constraints}
We wish to generalize the situation in which our graph-theoretic framework is used to witness quantum coherence.
There, vertices of an event graph $G$ are interpreted as representing vectors $\{\ket{\phi_i}\}_{i \in V(G)}$ in some Hilbert space $\mathcal{H}$, \ie pure quantum states. Edge weights then correspond to two-state quantum overlaps, $\vert \langle \phi_i \vert \phi_j \rangle \vert^2$. 
Such overlaps can be accessed empirically by \eg measuring one of the states on a measurement basis that includes the other.

Abstracting from this, we consider a situation in which we associate a preparation procedure $P_i$ to each vertex $i \in V(G)$ of a given graph $G$.
But in order to emulate the setup above for more general operational theories,
it is necessary to impose some additional operational constraints.
These constraints distil the aspects of quantum theory that make this work, 
allowing (a theory-independent version of) two-state overlaps.
We shall refer to them as the Lostaglio--Senno--Schmid--Spekkens (LSSS) operational constraints, after Refs.~\cite{Lostaglio2020contextualadvantage,schmid2018discrimination}.
Note that these constraints apply to preparation procedures; we need not assume any operational equivalences for measurement procedures. Therefore, the scenarios under consideration aim to probe preparation contextuality only.

First, for any preparation $P_i$, we assume that there is a corresponding `test measurement' $M_i$ with outcomes $\{0,1\}$
satisfying the operational statistics $p(M_i \mid P_i) = 1$.
In quantum theory, if $P_i$ is the preparation associated with state $\ket{\phi_i}$
then $M_i$ is realised by the projective measurement $\{\ket{\phi_i}\!\bra{\phi_i}, \mathbb{1} - \ket{\phi_i}\!\bra{\phi_i}\}$
where the first projector corresponds to the outcome $k=1$.

Moreover, for any edge $\{i,j\} \in E(G)$, whose incident vertices have preparations $P_i$ and $P_j$, 
we assume that there exists another pair of preparations $P_{i^\perp}$ and $P_{j^\perp}$ satisfying
$p(M_i \mid P_{i^\perp}) = 0$, $p(M_j \mid P_{j^\perp}) = 0$, and the operational equivalence
$\frac{1}{2}P_i + \frac{1}{2}P_{i^\perp} \simeq \frac{1}{2}P_j + \frac{1}{2}P_{j^\perp}$.
In quantum theory, this is always satisfied: given a pair of pure states $\ket{\phi_i}$ and $\ket{\phi_j}$,
one picks $\ket{\phi_{i^\perp}}$ to be the vector orthogonal to $\ket{\phi_i}$
living in the two-dimensional space spanned by $\{\ket{\phi_i},\ket{\phi_j}\}$, and similarly for $\ket{\phi_{j^\perp}}$.

The probabilities $p(M_i\mid P_j)$
are usually called the \stress{confusability}~\cite{lostaglio2020certifying,schmid2018discrimination}, because they may be interpreted as the probability of guessing incorrectly that the preparation performed had been $P_i$ instead of $P_j$.
These probabilities provide a theory-independent, operational treatment of two-state overlaps,
which reduces to the familiar notion in the case of quantum theory viewed as an operational theory:
\[p(M_i \mid P_j)  \stackrel{QT}{=} \Tr\left(\ket{\phi_i}\!\bra{\phi_i}\ket{\phi_j}\!\bra{\phi_j}\right) = \vert \langle \phi_i \vert \phi_j \rangle \vert^2 .\]
Therefore, we use these confusability probabilities to provide edge weights $r_{ij} = p(M_i \mid P_j)$ in our framework.
In summary, an assignment of preparation procedures to the vertices of $G$ such that the LSSS operational constraints are satisfied for the pairs of preparations associated to each edge determines an edge weighting $\fdec{r}{E(G)}{[0,1]}$.

\psection{Preparation noncontextuality}
When faced with an operational theory, a natural question is whether it admits a (noncontextual) hidden-variable explanation,
that is, whether it can be realised by a noncontextual ontological model.
In general, an ontological model consists of a measurable space $(\Lambda, \mathcal{F}_\Lambda)$ of \stress{ontic} states equipped with ontological interpretations for preparation and measurement procedures:
preparation procedures $P$ determine probability measures $\mu_P$ on $\Lambda$,
whereas measurement procedures $M$ determine measurable functions $\xi_M$ mapping each ontic state $\lambda \in \Lambda$ to (a distribution on) outcomes. Note that the interpretation of classical mixtures of procedures must be determined linearly from that of basic procedures, \eg $\mu_{\frac{1}{2}P+\frac{1}{2}Q} = \frac{1}{2}\mu_P+\frac{1}{2}\mu_Q$.
The composition of the interpretations of a preparation and a measurement (going via the ontic space $\Lambda$) is required to recover the empirical or operational predictions, \ie
\[p(\cdot \mid M, P) \;=\; \int_\Lambda \xi_M \, \mathrm{d} \mu_P ,\]
or with variables,
\[p( k \mid M, P) \;=\; \int_\Lambda \xi_M(k \mid \lambda) \, \mathrm{d} \mu_P(\lambda) .\]

Such a realization by an ontological model is said to be noncontextual if operationally equivalent procedures are given the same interpretation.
For preparations, the requirement is that two operationally equivalent preparation procedures determine the same probability measure on $\Lambda$.
We refrain from going into detail on the general definition, as the characterization that follows suffices.

In Refs.~\cite{Lostaglio2020contextualadvantage,schmid2018discrimination}, it was shown that the LSSS constraints imply that any preparation noncontextual model explaining preparation procedures $P_i$ as probability measures $\mu_i$ on $\Lambda$ must satisfy
\begin{equation}\label{equation: noncontextual overlaps}
    p(M_i\mid P_j) 
    = 1 - \|\mu_i - \mu_j\|_{_{\mathsf{TV}}} ,
\end{equation}
where $\|\cdot - \cdot\|_{_{\mathsf{TV}}}$ denotes the total variation distance between probability measures,
given for an arbitrary measurable space $(\Lambda, \mathcal{F}_\Lambda)$ by
\[\|\mu_i - \mu_j\|_{_{\mathsf{TV}}} = \sup_{E \in \mathcal{F}_\Lambda}|\mu_i(E) - \mu_j(E)|.\]
In the case when $\Lambda$ is discrete (which is effectively all we actually need),
this distance is related to the $l_1$ norm\footnote{{In the continuous case, it is often rendered as $\|\mu_i - \mu_j\|_{_{\mathsf{TV}}} = \int_\Lambda \vert \mu_i(\lambda) - \mu_j(\lambda) \vert \,\mathrm{d}\lambda$ in terms of a reference measure such as the Lebesgue measure on the real line.}}:
\begin{align*}
    \|\mu_i - \mu_j\|_{_{\mathsf{TV}}} = \frac{1}{2}\|\mu_i - \mu_j\|_{_1} = \frac{1}{2}\sum_{\lambda \in \Lambda}|\mu_i(\lambda)-\mu_j(\lambda)|.
\end{align*}

We can take that as a \textit{definition} of preparation noncontextual edge weightings. 

\begin{definition}\label{def:PNCedgeweighting}
Let $G$ be an event graph. An edge weighting $\fdec{r}{E(G)}{[0,1]}$ is said to be \stress{preparation noncontextual}
if
the edge weights are of the form in the right-hand side of \cref{equation: noncontextual overlaps}, \ie
 $r_{ij} = 1 - \|\mu_i - \mu_j\|_{_{\mathsf{TV}}}$,
for some choice of 
an (ontic) measurable space $\Lambda$ and of probability measures $\mu_i$ on $\Lambda$ for each vertex $i\in V(G)$.
\end{definition}

\psection{Cycle inequalities witness preparation contextuality}
We now show how in the case of cycle graphs the inequalities derived from our framework serve as witnesses of preparation contextuality for operational theories satisfying the LSSS constraints.

The technical result is stated in the following proposition; it follows from the triangle inequality.

\begin{proposition}
Any inequality bounding the set $C_{C_n}$ cannot be violated by a preparation noncontextual edge weighting (\cref{def:PNCedgeweighting}).
\end{proposition}
\begin{proof}
For simplicity, we use addition modulo $n$ when labelling the vertices of the cycle graph $C_n$, meaning that $i = i+n$.
From the triangle inequality of the norm $\| \cdot \|_{_{\mathsf{TV}}}$ it follows that
\begin{align*}
    & \| \mu_i - \mu_{i+n-1}\|_{_{\mathsf{TV}}}  \\
    = \; & \| \mu_i \underbrace{-\mu_{i+1}+\mu_{i+1}-\dots -\mu_{i+n-2}+\mu_{i+n-2}}_{n-2\text{ zeros}}-\mu_{i+n-1}\|_{_{\mathsf{TV}}}\\
    \leq \; & \| \mu_i - \mu_{i+1}\Vert_{_{\mathsf{TV}}} + \dots + \Vert \mu_{i+n-2}-\mu_{i+n-1}\|_{_{\mathsf{TV}}} .
\end{align*}
Therefore, writing $\|\mu_{i,j}\|_{_{\mathsf{TV}}} \defeq \|\mu_{i}-\mu_j\|_{_{\mathsf{TV}}}$ for clarity,
\[
    \Vert \mu_{i,i+n-1}\Vert_{_{\mathsf{TV}}} - \Vert \mu_{i,i+1}\Vert_{_{\mathsf{TV}}} - \dots - \Vert \mu_{i+n-2,i+n-1}\Vert_{_{\mathsf{TV}}} \leq 0 .
\]
We must now add $1$ to each term to recover the noncontextual overlaps of \cref{equation: noncontextual overlaps}.
We have $n$ terms, but since the first term has a different sign, two of these $1$s will cancel, leaving $n-2$ added to both sides of the inequality:
\begin{align*}
    -1+\Vert \mu_{i,i+n-1}\Vert_{_{\mathsf{TV}}} +1-\Vert \mu_{i,i+1}\Vert_{_{\mathsf{TV}}}  &\\
    +\;\cdots\;  + 1 - \Vert \mu_{i+n-2,i+n-1} 
    \Vert_{_{\mathsf{TV}}}  
    & \;\;\leq\;\; n-2.
\end{align*}
Recalling that $r_{ij} = 1-\Vert \mu_{i,j} \Vert_{_{\mathsf{TV}}}$, we 
recover a cycle inequality for any chosen vertex $i$:
\[
    -r_{i,i+n-1}+r_{i,i+1}+\dots+r_{i+n-2,i+n-1}\leq n-2 .
\]
\end{proof}

We may see this result from two perspectives. We can take a \textit{theory-dependent perspective} and look for what information we can extract assuming quantum theory as the relevant operational theory; this proposition then shows that pure quantum states that violate the $n$-cycle inequalities can be used to construct a proof of quantum preparation contextuality. The construction is done by constructing states and measurements that represent a realization of the prepare-and-measure scenario described by the LSSS constraints. In summary, violations of these inequalities serve as \textit{witnesses of quantum preparation contextuality}.

In light of this result, the experiment of Ref.~\cite{Giordani21} can be understood as an experimental test that witnessed  preparation contextuality of quantum theory; however since the purpose was not to witness preparation contextuality the authors have not experimentally probed the relevant operational equivalences, and have not ruled out loopholes for such a test.

We can also take a \textit{theory-independent perspective}.
If a given operational theory satisfying the LSSS constraints for some cycle graph admits a preparation noncontextual ontological model, then the confusabilities $r_{ij} = p(M_i\vert P_j)$ are bounded by the cycle inequalities.
For instance, the Spekkens Toy Theory~\cite{spekkens2007evidence} satisfies the LSSS constraints for any pair of preparation procedures. Since it admits a noncontextual ontological model, it cannot violate the cycle inequalities.

\bibliography{cohe_context}
\end{document}

%% file: macros.tex
\usepackage[british]{babel}
\usepackage{amsmath,amsthm,amssymb,amsfonts}
\usepackage{stmaryrd}

\usepackage{microtype}
\usepackage{url}
\usepackage[unicode]{hyperref}
\usepackage{cleveref}
  \crefformat{equation}{Eq.~(#2#1#3)} 
  \Crefformat{equation}{Equation~(#2#1#3)}
  \crefformat{figure}{Fig.~#2#1#3}
  \crefrangeformat{equation}{Eqs.~#3(#1)#4--#5(#2)#6}

\usepackage[usenames,dvipsnames]{xcolor}
\usepackage[all]{xy}
\usepackage{tikz-cd}
\usepackage{graphicx}

\allowdisplaybreaks[1]

\usepackage{braket}
\usepackage{xfrac}
\usepackage{float}

\def\stress{\textit}
\newcommand{\ie}{i.e.~}
\newcommand{\eg}{e.g.~}

\theoremstyle{plain}
\newtheorem{theorem}             {Theorem}
\newtheorem{proposition}[theorem]{Proposition}
\newtheorem{lemma}      [theorem]{Lemma}
\newtheorem{corollary}  [theorem]{Corollary}

\newtheorem*{theorem*}    {Theorem}
\newtheorem*{proposition*}{Proposition}
\newtheorem*{lemma*}      {Lemma}
\newtheorem*{corollary*}  {Corollary}
\newtheorem*{conjecture*} {Conjecture}

\theoremstyle{definition}
\newtheorem{definition}[theorem]{Definition}

\newtheorem*{definition*}{Definition}
\newtheorem*{example*}   {Example}
\newtheorem*{question*}   {Question}
\newtheorem*{idea*}  {Idea}

\theoremstyle{remark}

\newcommand{\Mdot}{\text{ .}}

\usepackage{colonequals}
\newcommand{\defeq}{\colonequals} 
\newcommand{\setdef}  [2]{\left\{#1 \mid #2\right\}}             
\newcommand{\enset}   [1]{\mathopen{ \{ }#1\mathclose{ \} }} 
\renewcommand{\to}{\longrightarrow}
\newcommand{\fdec}    [3]{#1 \colon #2 \longrightarrow #3}

\newcommand{\Forall}[1]{\forall {#1}\boldsymbol{.}\;}

\newcommand{\es}{\varnothing}

\newcommand{\NN}{\mathbb{N}}

\newcommand{\RR}{\mathbb{R}}

\newcommand{\vc}[1]{\mathbf{#1}}

\newcommand{\Galpha}{{G / \alpha}}

\newcommand{\Eq}[1]{\mathsf{Eq}(#1)}
\newcommand{\Eqk}[2]{\mathsf{Eq}_{#2}(#1)}
\newcommand{\EqG}{\Eq{G}}
\newcommand{\EqGk}{\Eqk{G}{k}}
\newcommand{\EqGprime}{\Eq{G'}}

\newcommand{\Tr}{\mathrm{Tr}}
\newcommand{\STAB}{\mathrm{STAB}}
\newcommand{\chiS}{\chi_{_{\scriptsize S}}}
\newcommand{\chiW}{\chi_{_{\scriptsize W}}}

\newcommand{\algtab}{\hspace{1.5em}}